\documentclass[pra,onecolumn,superscriptaddress]{article}
\usepackage{amsmath,amsfonts,amsthm,amssymb}
\usepackage{braket}
\usepackage{graphicx}
\usepackage{algorithm}
\usepackage{tabto}
\usepackage{fullpage}
\usepackage[noadjust]{cite}
\usepackage[dvipsnames]{xcolor}
\usepackage{comment}
\usepackage{footnote}
\makesavenoteenv{algorithm}
\usepackage[backref=page,colorlinks,citecolor=blue,bookmarks=true,final]{hyperref}
\newtheorem{definition}{Definition}
\usepackage[bottom]{footmisc}
\usepackage{thm-restate}
\newtheorem{theorem}{Theorem}[section]
\newtheorem{corollary}[theorem]{Corollary}
\newtheorem{lemma}[theorem]{Lemma}

\newtheorem{remark}[theorem]{Remark}
\newtheorem{fact}[theorem]{Fact}

\newcommand{\PRS}{\textnormal{PRS}}
\newcommand\abs[1]{\left|#1\right|}
\newcommand{\negl}{\textnormal{negl}}
\newcommand{\skgen}{\textit{SKGen}}
\newcommand{\pkgen}{\textit{PKGen}}
\newcommand{\sign}{\textit{Sign}}
\newcommand{\verify}{\textit{Verify}}

\title{On black-box separations of quantum digital signatures\\ from pseudorandom states}
\author{Andrea Coladangelo\thanks{Paul G. Allen School of Computer Science and Engineering, University of Washington. Email: \texttt{coladan@cs.washington.edu}.} \and Saachi Mutreja\thanks{Columbia University. Email: \texttt{saachi@berkeley.edu}.}}

\begin{document}

\maketitle
\begin{abstract}
It is well-known that digital signatures can be constructed from one-way functions in a black-box way. While one-way functions are essentially the minimal assumption in classical cryptography, this is not the case in the quantum setting. A variety of qualitatively weaker and \emph{inherently quantum} assumptions (e.g.\ EFI pairs, one-way state generators, and pseudorandom states) are known to be sufficient for non-trivial quantum cryptography.

While it is known that commitments, zero-knowledge proofs, and even multiparty computation can be constructed from these assumptions, it has remained an open question whether the same is true for quantum digital signatures schemes (QDS). In this work, we show that there \emph{does not} exist a black-box construction of a QDS scheme \emph{with classical signatures} from pseudorandom states with linear, or greater, output length. Our result complements that of Morimae and Yamakawa (2022), who described a \emph{one-time} secure QDS scheme with classical signatures, but left open the question of constructing a standard \emph{multi-time} secure one.
\end{abstract}
\newpage
{\small\tableofcontents}
\newpage

\section{Introduction}
One of the foundational goals of cryptography is to study the \emph{minimal assumptions} needed to construct cryptographic functionalities of interest. While the existence of one-way functions is generally considered to be the minimal assumption that is useful for cryptography, a recent line of work, initiated by Kretschmer \cite{kretschmer2021quantum}, has shown that this may not be the case \emph{in a quantum world}. Since Kretschmer's result, the topic has seen a surge of interest, with many recent works constructing cryptography from assumptions that are potentially weaker than one-way functions \cite{ananth2022cryptography, morimae2022quantum, morimae2022one, ananth2023pseudorandom, ananth2023pseudorandomStrings, khurana2023commitments}.

These constructions are based on novel primitives whose security is formulated in terms of the hardness of an \emph{inherently quantum} problem. The first example of such a primitive, a pseudorandom state (PRS), was proposed by Ji, Liu, and Song \cite{ji2018pseudorandom}. A PRS can be thought of as the quantum analogue of a pseudorandom generator (PRG), and it refers to an ensemble of efficiently preparable quantum states that are \emph{computationally} indistinguishable from Haar random. Other more recent examples are EFI pairs \cite{brakerski2022computational}, and one-way state generators \cite{morimae2022one}. These inherently quantum primitives are sometimes collectively referred to as ``\emph{MicroCrypt}''\footnote{
\emph{MicroCrypt} is an addition to Impagliazzo's five worlds \cite{impagliazzo1995personal}. This is a world in which one-way functions do not exist, but inherently quantum primitives, like PRS, exist, and thus non-trivial cryptography is possible. As far as we know, the term was coined by Tomoyuki Morimae.}. They are especially interesting for two reasons. First, they are qualitatively weaker than one-way functions: while one-way functions imply all of them, Kretschmer showed that they are provably not sufficient to construct one-way functions when used in a black-box way \cite{kretschmer2021quantum}. Second, these primitives are sufficient to construct many cryptographic primitives of interest, namely commitments, zero-knowledge proofs, symmetric-key encryption, and even oblivious transfer and multiparty computation \cite{ananth2022cryptography, morimae2022quantum, morimae2022one, ananth2023pseudorandom, ananth2023pseudorandomStrings, khurana2023commitments}. These results are counterintuitive because in a classical world all of these primitives require, at the very least, one-way functions! The new constructions circumvent this requirement because some component of the construction involves quantum states, e.g.\ the communication in the case of commitments or oblivious transfer, and the ciphertext in the case of symmetric-key encryption. In light of this, what more can we hope to construct in MicroCrypt, and what is beyond reach? In the rest of this introduction, we focus on PRS, as they imply all other known primitives in MicroCrypt.

One of the most important primitives whose relationship to MicroCrypt is still elusive are digital signatures. While classical digital signatures imply one-way functions, and are thus beyond the reach of MicroCrypt, recent works have explored the possibility of constructing digital signature schemes where the public key is a quantum state. In particular, Morimae and Yamakawa \cite{morimae2022one} construct a \emph{one-time} secure digital signature scheme with quantum public keys from pseudo random states in a black-box way. One-time security means that the adversary is only allowed to make \emph{one} query to the signing oracle, before attempting to produce a valid signature of a different message. Morimae and Yamakawa's construction is a quantum public-key version of the classical Lamport signature scheme. 
However, it is unclear how to extend their construction to satisfy the standard notion of \emph{multi-time} security (where the adversary is allowed an arbitrary polynomial number of queries to the signing oracle). One of the main obstacles is that the public keys are quantum states (the classical approach involves signing public keys, and signing quantum states is known to be impossible in general \cite{alagic2021can}). Morimae and Yamakawa thus leave open the question of whether there exists a black-box construction of a quantum digital signature scheme with multi-time security from a PRS. This is the central question that we focus on in this work.
\begin{center}
\emph{Does there exist a black-box construction of multi-time secure quantum digital signatures from PRS?} 
\end{center}
On the one hand, the barrier in extending Morimae and Yamakawa's scheme seems fundamental, and coming up with entirely new schemes is always a difficult endeavour. On the other hand, Kretschmer's original separation of PRS from one-way functions \cite{kretschmer2021quantum} is the only black-box separation involving MicroCrypt that we are aware of, and thus known techniques are fairly limited. 

\subsection{Our results}
We provide a partial answer to the above question on the negative side. Namely, we show the following.

\begin{theorem}[Informal]
There is a quantum oracle $O$ relative to which:
\begin{enumerate} 
    \item PRS with linear, or greater, output length exist.
    \item No digital signature scheme with a \emph{quantum} public key, and classical secret key and signatures, exists.\footnote{Our separation rules out QDS schemes where the message length is at least logarithmic in the security parameter. For QDS schemes where the message length is shorter (e.g.\ $1$-bit messages), we rule out QDS schemes under the notion of multi-time security where the adversary succeeds as long as it produces \emph{any} valid signature that has not been seen before (even of a message that was previously queried). See remark \ref{rem:technicality1} for more details.}
    \end{enumerate}
\end{theorem}

The oracle is similar to the one that Kretschmer used to separate PRS and one-way functions \cite{kretschmer2021quantum}. Our analysis builds on a key technique introduced in the same work, but 
needs to circumvent several additional roadblocks that we discuss in the technical overview (Section \ref{sec:techoverview}). We believe that these roadblocks may not be unique to this setting, and that our proof ideas (described in Section \ref{sec:techoverview2}) might find application elsewhere. As a corollary, our result implies the following.

\begin{corollary}[Informal] 
\label{cor:1}
There does not exist a fully black-box construction of a digital signature scheme with a quantum public key (and classical secret key and signatures) from a PRS with linear, or greater, output length.\footnote{In fact, our result is slightly stronger than this. We rule out a fully black-box construction where secret keys can also be quantum states, as long as the secret key generation algorithm does not make queries to the PRS generation algorithm. See Remark \ref{rem:technicality2} for the details.}
\end{corollary}

To the best of our knowledge, this is the first non-trivial black-box separation involving MicroCrypt beyond Kretschmer's original separation. We point out that our separation is in the same setting as Morimae and Yamakawa's positive result (PRS of linear output length, and classical secret key and signatures). Thus it directly answers, in the negative, their question of whether their approach could be extended to yield multi-time security.


How should one view a black-box separation? The vast majority of known cryptographic constructions are \emph{fully black-box}. This means that:
\begin{itemize}
\item[(i)] The construction of primitive $\mathcal{Q}$ from primitive $\mathcal{P}$ does not make use of the ``code'' of $\mathcal{P}$, but only uses $\mathcal{P}$ as a black-box.
\item[(ii)] There exists a black-box security reduction, i.e.\ given an adversary $A$ that breaks $\mathcal{Q}$, there exists an adversary $A'$ that breaks $\mathcal{P}$ by using $A$ as a black-box.
\end{itemize}
Showing that such a fully black-box reduction does not exist rules out the most natural class of constructions, and establishes that any attempt to construct $\mathcal{Q}$ from $\mathcal{P}$ \emph{must} violate either (i) or (ii). From the point of view of ``cryptographic complexity'', the black-box separation establishes that primitive $\mathcal{Q}$ is, at the very least, not qualitatively weaker than primitive $\mathcal{P}$, and possibly stronger.

The approach of exhibiting an oracle separation as a means to prove the impossibility of a black-box construction was introduced in a seminal work of Impagliazzo and Rudich \cite{impagliazzo1989limits}. In Section \ref{sec:black-box-constructions}, we include a formal discussion of the relationship between oracle separations and black-box constructions \emph{in a quantum world}, i.e.\ a world in which oracles are unitary and constructions are quantum algorithms. Such a discussion, to the best of our knowledge, was missing despite recent works on the topic. The summary is that, when talking about a black-box construction of primitive $\mathcal{Q}$ from primitive $\mathcal{P}$ in a quantum world, one needs to be careful about defining the kind of ``access to $\mathcal{P}$'' that is available to the construction. One natural definition is that ``access to $\mathcal{P}$'' means having access to a ``unitary implementation'' of $\mathcal{P}$. However, a natural question is: is access to the ``inverse'' also available? Perhaps unsurprisingly, if one wants to rule out a black-box construction \emph{with access to the inverse}, then one should exhibit a separation relative to a pair of oracles $(O,O^{-1})$. 

We point out that our result (Corollary \ref{cor:1}) only rules out a fully black-box construction \emph{without} access to the inverse. This limitation is shared by Kretschmer's separation of PRS and one-way functions (this is not by coincidence, but it is rather because our result leverages some of techniques used there).

\subsection{Open questions}
Our result comes short of a full answer to the general question of the relationship between digital signatures and PRS in two respects:
\begin{enumerate}
\item First, our result only rules out a black-box construction of digital signatures from PRS with long output (linear or greater). However, can digital signatures be constructed from PRS with short output (sublinear)? We should point out that, unlike for classical PRGs, for which the output can be stretched, and also (trivially) shrunk, the relationship between PRS with short and long output is still very much unclear. We do not know whether one implies the other (and if so in what direction), or whether they are incomparable. Recent work \cite{ananth2023pseudorandomStrings} shows that PRS with short output can be used to construct primitives (e.g.\ QPRGs) that we do not know how to construct from PRS with long output. So it seems at least in principle possible that there could be a black-box construction of digital signatures from PRS with short output.
\item Second, our result only applies to digital signatures with a quantum public key, but with classical secret key and signatures. If we allow the latter to be quantum as well, then is there a construction? On the one hand, it is unclear to us how this relaxation may be helpful in realizing a construction. On the other, our current techniques to prove a separation run into a barrier in this setting, which we discuss in the technical overview, and in further detail in Section \ref{sec:barrier}.
\end{enumerate}

\section*{Acknowledgements}
We thank Fermi Ma for several helpful discussions. We also thank Or Sattath for drawing our attention to the fact that our result only separates quantum digital signatures from PRS of linear, or greater, output length, and that a black-box construction may exist from PRS with shorter output length.
SM is supported by AFOSR award FA9550-21-1-0040, NSF CAREER award CCF-2144219, and the Sloan Foundation. This work was partially completed while the authors were visiting the Simons Institute for the Theory of
Computing.

\section{Technical Overview}
\label{sec:techoverview}
We give a detailed informal overview of our result that PRS with linear, or greater, output length, cannot be used to construct, in a black-box way, a digital signature scheme with quantum public key, and classical secret key and signatures.
For convenience, from here on, we simply refer to the latter type of scheme as a QDS.

To show this result, it is sufficient to construct an oracle (classical or quantum) relative to which PRS exist but QDS do not (see Section \ref{sec:black-box-constructions} for more details about why this is sufficient).
Before explaining our approach, we give a (slightly informal) definition  of a QDS scheme and its security. 

A QDS scheme is specified by a tuple of algorithms $(\skgen, \pkgen, \sign, Verify)$ satisfying the following:
\begin{itemize}
    \item $\skgen(1^{\lambda}) \rightarrow sk$: is a QPT algorithm that takes as input  $1^{\lambda}$, and outputs a classical secret key $sk$.
    \item $\pkgen(\textit{sk}) \rightarrow \ket{pk}$: is a deterministic algorithm that takes as input a secret key $\textit{sk}$, and outputs the quantum state $\ket{pk}$.\footnote{To clarify, $\ket{pk}$ is allowed to be an arbitrary pure state (not necessarily a standard basis state).} We additionally require $\ket{pk}$ to be fixed given $sk$, i.e.\ the algorithm $\skgen$ is deterministic (it consists of a fixed unitary quantum circuit acting on the input $sk$, and some auxiliary registers).\footnote{We include this requirement so that the notion of ``quantum'' public key is a little more faithful to the spirit of a classical public key. This requirement ensures that the party in possession of the secret key can generate multiple copies of the corresponding public key. Note that for a completely classical digital signature scheme this requirement is without loss of generality, since any randomness used in the generation procedure can be included in the secret key.}
      
    \item $\sign(\textit{sk,m}) \rightarrow \sigma$: is a QPT algorithm that takes as input a secret key \textit{sk} and a classical message \textit{m}, and outputs a classical signature $\sigma$.

    \item $\verify(\ket{pk},m,\sigma) \rightarrow \mathsf{accept}/\mathsf{reject}$: is a QPT algorithm that takes as input a public key $\ket{pk}$, a message $m$, and a candidate signature $\sigma$, and outputs $\mathsf{accept}$ or $\mathsf{reject}$.
\end{itemize}
We take messages and the secret key to be of length $\lambda$ for simplicity. In this work, we focus on standard \emph{multi-time} security, defined in terms of the following ``unforgeability'' game between an adversary $\mathcal{A}$ and a challenger $\mathcal{C}$.
\begin{itemize}
    \item[(i)] $\mathcal{C}$ samples $sk \leftarrow \skgen(1^{\lambda})$, and $\mathcal{A}$ receives polynomially many copies of $\ket{pk} = \pkgen(\textit{sk})$.
    \item[(ii)] $\mathcal{A}$ obtains from $\mathcal{C}$ the signatures of polynomially many messages of its choice.
    \item[(iii)] $\mathcal{A}$ sends a pair $(m,\sigma)$ to $\mathcal{C}$, where $m$ is not among the previously signed messages.
    \item[(iv)] $\mathcal{A}$ wins the game if $\verify(pk,m,\sigma)$ accepts.
\end{itemize}
The QDS scheme is \emph{multi-time} secure if any quantum polynomial time adversary has negligible winning probability in this game.

\paragraph{The separating oracle.}
We are now ready to describe the oracle relative to which PRS exist, but QDS schemes of type (1) does not. As mentioned earlier, the oracle is similar to the one used by Kretschmer in \cite{kretschmer2021quantum}. The oracle $O$ consists of a pair of oracles $(\mathcal{U}, \mathcal{Q})$, where $\mathcal{Q}$ is a classical oracle solving a fixed $\mathsf{EXP}$-complete problem, and $\mathcal{U}$ is a collection of Haar-random unitaries $\{\mathcal{U}_{\ell}\}_{\ell \in \mathbb{N}}$, where each $\mathcal{U_{\ell}}$ is an indexed list of $2^{\ell}$ Haar-random unitaries acting on $\ell$ qubits. 

\paragraph{Why is $(\mathcal{U},\mathcal{Q})$ a natural choice of oracle?}
First of all, relative to $\mathcal{U}$, there is a trivial construction of a PRS: on input a seed $k$, apply the unitary from $\mathcal{U}_{|k|}$ with index $k$ to the state $\ket{0}^{\otimes |k|}$. Importantly, this construction is still secure even in the presence of the second oracle $\mathcal{Q}$, provided $\mathcal{Q}$ is fixed independently of the sampled $\mathcal{U}$ (this is shown in Lemma 31 of \cite{kretschmer2021quantum}).

 The hard part of our result is constructing, for any QDS scheme relative to these oracles, an adversary $\mathcal{A}^{\mathcal{U},\mathcal{Q}}$ that breaks it. In this technical overview, as a warm-up, we start by considering the case of a QDS scheme with a \emph{classical} public key (and classical signatures). In this case, a simple approach suffices: $\mathcal{A}$ uses $\mathcal{Q}$ to perform a ``brute-force'' search for a signature that passes the verification procedure (this ``brute-force'' search can be performed because $\mathcal{Q}$ solves an $\mathsf{EXP}$-complete problem). Since $\mathcal{Q}$ is a function that is fixed \emph{before} $\mathcal{U}$ is sampled (informally speaking, $\mathcal{Q}$ does not have access to $\mathcal{U}$), the brute-force search approach has to be combined with a technique introduced by Kretschmer \cite{kretschmer2021quantum} to simulate the queries that the verification procedure makes to $\mathcal{U}$. We then move on to the case of a QDS scheme with a \emph{quantum} public key (and classical signatures). Here, there are several challenges that prevent us from using the same ``brute-force'' search approach. We describe these challenges, and an approach that overcomes them.

\subsection{Warm up: oracle separation between PRS and QDS with \emph{classical} public key}
\label{sec:techoverview1}
Consider a QDS scheme $(\skgen^{\mathcal{U},\mathcal{Q}}, \pkgen^{\mathcal{U},\mathcal{Q}}, \sign^{\mathcal{U},\mathcal{Q}}, Verify^{\mathcal{U},\mathcal{Q}})$ with a classical public key. Unless specified otherwise, we will always consider schemes with classical secret keys and signatures. Our task is to construct an adversary $\mathcal{A}^{\mathcal{U},\mathcal{Q}}$ that wins the multi-time unforgeability game with non-negligible probability.
 
\paragraph{How does $\mathcal{A}$ use $\mathcal{Q}$?}
For concreteness, suppose signatures for some message $m$ are $\lambda$-bit strings. Let $pk$ be the public key. The simplest approach to finding a valid signature for message $m$, is to do a ``brute-force search'' over the space of $\lambda$-bit strings for a string that is accepted by $\verify(pk, \cdot)$. One needs to be a bit more careful though since $\verify$ makes queries to both $\mathcal{U}$ and $\mathcal{Q}$. Ignoring $\mathcal{U}$ for a moment, since $\mathcal{Q}$ solves a fixed $\textsf{EXP}$-complete problem, the brute-force search naively corresponds to an $\mathsf{EXP}^{\mathsf{EXP}}$ problem (which may thus be outside of $\mathsf{EXP}$). However, since $\verify$ is a $QPT$ algorithm (and thus can only make polynomial-size queries to $\mathcal{Q}$), this brute-force search actually corresponds to an $\mathsf{EXP}$ problem (and can thus be reduced to an instance of the $\mathsf{EXP}$-complete problem solved by $\mathcal{Q}$)\footnote{For more details, see the second footnote within Algorithm~\ref{algo:findsk}}.

The more delicate issue is that, while the algorithm $\verify$ has a succinct description, unfortunately the oracle $\mathcal{U}$ does not. Moreover, the oracle $\mathcal{U}$ is sampled \emph{after} $\mathcal{Q}$ is fixed, and so even an inefficient description of $\mathcal{U}$ cannot be ``hardcoded'' into $\mathcal{Q}$. As anticipated, this issue can be resolved using a technique by Kretschmer \cite{kretschmer2021quantum}, which we describe below.\footnote{As an alternative to this classical oracle $\mathcal{Q}$, one might also consider a quantum oracle $\mathcal{Q}$ that makes queries to $\mathcal{U}$. Given exponentially many queries to $\mathcal{U}$, such an oracle would allow the adversary to easily break the QDS scheme. However, such an oracle would also break the security of the PRS built using $\mathcal{U}$.} This approach still runs into several fundamental issues when the public key is \emph{quantum}. We describe these issues, and how to overcome them, in Section \ref{sec:techoverview2}, but for now we focus on the case of a \emph{classical} public key.

\begin{remark}
For simplicity, in the rest of this section, we will often describe $\mathcal{Q}$ as an exponential-time algorithm. Formally, however, $\mathcal{Q}$ is a fixed function computing a fixed \textsf{EXP}-complete problem. So, when describing $\mathcal{Q}$ as an exponential-time algorithm what we formally mean is: we cast the underlying problem that the algorithm is solving as an $\mathsf{EXP}$ problem, and then reduce it to the particular \textsf{EXP}-complete problem computed by $\mathcal{Q}$.
\end{remark}

\paragraph{Simulating queries to $\mathcal{U}$.}
The technique relies crucially on the strong concentration property of the Haar measure. The corollary of this property that is most relevant here is the following (stated informally). Let $C$ be a quantum circuit that makes $poly(\lambda)$ queries to a Haar random unitary acting on $\lambda$ qubits. Then, with overwhelming probability over (independently) sampling two such Haar random unitaries $U$ and $U'$, the output distributions of the circuits $C^U$ and $C^{U'}$ (on, say, the $\ket{0}$ input) are within a small constant TV distance of each other. In fact, the concentration is strong enough to support a union bound over \emph{all} standard basis inputs. So, with overwhelming probability over $U$ and $U'$, the output distributions of $C^U$ and $C^{U'}$ on \emph{all} standard basis inputs, are within a small constant TV distance of each other. 

In our setting, $C$ is the circuit $Verify(pk,m, \cdot)$ for some message $m$, which makes $T$ queries to the family $\mathcal{U}$. Thanks to the above concentration property, $\mathcal{Q}$ can now perform the brute-force search ``without access to $\mathcal{U}$'' by simply replacing the oracle calls of $Verify(pk,m, \cdot)$ with unitary $T$-designs. This will perfectly simulate $T$ queries to freshly sampled Haar random unitaries. 

There is still one remaining subtlety. Recall that $\mathcal{U}$ is a \emph{family} of unitaries $\{\mathcal{U}_{\ell}\}_{\ell \in \mathbb{N}}$, where each $\mathcal{U_{\ell}}$ is an indexed list of $2^{\ell}$ Haar-random unitaries acting on $\ell$ qubits. However, the concentration property only holds for Haar random unitaries acting on a \emph{large enough number of qubits}. This issue can be circumvented because for smaller dimensions, up to $O(\log(\lambda))$ qubits, the unitaries can be ``learnt'' efficiently (in $O(poly(\lambda))$ queries) by performing process tomography.

We can package Kretschmer's technique into one procedure, which we will denote as $\textsf{Sim-Haar}$. The latter procedure has two parameters $\eta$ and $\delta$. It has oracle access to $\mathcal{U}$, and takes as input the description of a quantum circuit $C$ (with a one-bit output) that makes queries to $\mathcal{U}$, and it outputs another quantum circuit $C'$ (with a one-bit output) that \emph{does not} make queries to $\mathcal{U}$. The guarantee of $\textsf{Sim-Haar}$ is that, with probability $1-e^{-\eta}$ over $\mathcal{U}$ and the randomness of the procedure, it holds that, for a given $x$, 
$$\Big|\Pr[C^{\mathcal{U}}(x) = 1] - \Pr[C'(\ket{x}) = 1]\Big| \leq \delta \,.$$
The runtime of $\textsf{Sim-Haar}$ (which includes queries to $\mathcal{U}$) is $poly(|C|, T, \eta, 1/\delta)$, where $|C|$ is the size of $C$.

To put things together, the adversary $\mathcal{A}^{\mathcal{U}, \mathcal{Q}}$ for the QDS scheme picks an arbitrary message $m$. It first obtains a circuit $Verify'$ by running $\textsf{Sim-Haar}^{\mathcal{U}}$ on input $Verify(pk, m, \cdot)$ (with a sufficiently large $\eta$, and a small constant $\delta$). Then, it invokes $\mathcal{Q}$ to search for a signature $\sigma$ such that $Verify'(\sigma)$ accepts with probability greater than a sufficiently large constant, and finally it outputs $(m, \sigma)$.

\subsection{Oracle Separation between PRS and QDS with \emph{quantum} public key}
\label{sec:techoverview2}

There are several issues when trying to extend the previous `brute force' attack to a QDS scheme with a quantum public key.

\paragraph{The first issue with a quantum $pk$.} Let the public key be a quantum state $\ket{pk}$. The first issue is syntactical: since $\mathcal{Q}$ is a classical oracle, how does $\mathcal{A}$ describe the circuit $Verify(\ket{pk}, \cdot)$ to $\mathcal{Q}$? The natural way to fix this is to consider a \emph{quantum} oracle $\mathcal{Q}$ (i.e.\ one that can take quantum inputs) that is still independent of $\mathcal{U}$. However, even then, $\mathcal{Q}$ would in general need exponentially many copies of $\ket{pk}$ in order to run a brute force search algorithm. This is because the public key state is potentially disturbed every time the circuit $Verify$ is run.\footnote{Note that attempts to ``uncompute'' the circuit and recover $\ket{pk}$ fail in general for several reasons. One of them is that correctness and soundness are not perfect, so negligible errors can add up to a noticeable quantity when performing a brute force search.}


\paragraph{An alternative approach:} Note that we should have expected the previous approach to fail. This is because it did not make use of the adversary's ability to make queries to the signing oracle. Since there exists a black-box construction of a one-query secure QDS scheme from a PRS \cite{morimae2022quantum}, any adversary breaking the QDS scheme must necessarily make use of signing queries (and in fact polynomially many of them). We consider the following alternative approach. 

\begin{itemize}
\item $\mathcal{A}$ makes polynomially many queries to the signing oracle, obtaining message-signature pairs $(m_i, \sigma_i)$.
\item $\mathcal{A}$ uses $\mathcal{Q}$ to find a set of secret keys that are ``consistent'' with all of the $(m_i, \sigma_i)$. We will refer to this set as $Consistent$. That is, $sk \in Consistent$ if, for all $(m_i, \sigma_i)$,
$\Pr[Verify^{\mathcal{U, Q}}(\pkgen^{\mathcal{U,Q}}(sk), m_i, \sigma_i) = \textsf{accept}]$ is greater than some treshold, for example $\frac{9}{10}$. Once again, just like in Subsection \ref{sec:techoverview1}, $\mathcal{Q}$ is independent of $\mathcal{U}$, and so we need to first obtained a simulated version of the circuit \\
$Verify^{\mathcal{U, Q}}(\pkgen^{\mathcal{U,Q}}(\cdot), \cdot, \cdot)$ using the \textsf{Sim-Haar} procedure described in Subsection \ref{sec:techoverview1}. Recall that the guarantee of \textsf{Sim-Haar} is that, with high probability over $\mathcal{U}$, on \emph{all} classical inputs $sk, m, \sigma$ the simulated circuit's acceptance probability is close to the original. 

For the rest of the section, whenever we refer to a circuit that originally made queries to $\mathcal{U}$, we will simply assume that we are utilizing a simulated version of that circuit obtained using \textsf{Sim-Haar}, and we will drop $\mathcal{U}$ from the notation. For ease of notation, since $\mathcal{Q}$ is fixed, we will also drop $\mathcal{Q}$ from the notation.
\item $\mathcal{A}$ signs a fresh message using a uniformly random key from $Consistent$.
\end{itemize}

\paragraph{The main challenge:}
Observe that the set of consistent secret keys are, by definition, those $sk$ such that, for all $(m_i, \sigma_i)$, $\Pr[Verify(\pkgen(sk), \cdot) = \mathsf{accept}] > 9/10$. By a suitable concentration bound (taking the number of queried message-signature pairs to be a large enough polynomial), we have that, with overwhelming probability over the $(m_i, \sigma_i)$, such $sk$ are also such that, for most $m$,
\begin{equation}
\Pr[Verify(\pkgen(sk), m, \sigma)  =\mathsf{accept}] \geq \Omega(1) \,, \textnormal{ where } \sigma \leftarrow \sign(sk^*, m) \,. \label{eq:reverse good signer}
\end{equation}
In other words, the secret keys in $Consistent$ accept (with high probability) not only the queried message, signature pairs, but also \emph{fresh} signatures signed using the true secret key $sk^*$. 

Unfortunately, this is not quite the guarantee we are looking for! What we want is the reverse (i.e.\ we would like the roles of $sk$ and $sk^*$ to be swapped): an $sk$ such that, for most $m$,
\begin{equation}
\label{eq:goodsigner}
\Pr[Verify(\pkgen(sk^*), m, \sigma)  =\mathsf{accept}] \geq \Omega(1) \,,\textnormal{ where } \sigma \leftarrow \sign(sk, m) \,.  
\end{equation}
 To make the issue concrete, consider an $sk$ for which $Verify(\pkgen(sk), \cdot, \cdot)$ simply accepts everything with probability $1$. Such an $sk$ would certainly be in the consistent set of secret keys, but it may not be capable of generating signatures that are accepted by the true secret key $sk^*$ (and the existence of such an $sk$ does not appear to contradict any property of a QDS scheme).

To summarize, so far we have identified a consistent set of secret keys that clearly contains $sk^*$, but is potentially exponentially large, and may contain secret keys that do not satisfy Equation \ref{eq:goodsigner}, i.e. they are not ``good signers''. So, how do we proceed?

\paragraph{Finding a sequence of smaller and smaller subsets containing $sk^*$.} We describe an iterative procedure that identifies smaller and smaller subsets of $Consistent$, which are \emph{guaranteed to still contain} $sk^*$ (or else there is an easy way to find a signature accepted by $\pkgen(sk^*)$). Eventually, these subsets only contain the true secret key $sk^*$.

Before describing the iterative procedure, we define the following terms:
\begin{itemize}
\item \emph{``Good signer'' secret keys}: Informally, a secret key $sk \in Consistent$ is a ``good signer'' if $\frac{9}{10}$ of the other secret keys in $Consistent$ accept signatures generated using $sk$ on a constant fraction of the message space, with constant probability. More precisely,
$sk \in Consistent$ is a good signer iff the following is true: $|accept_{sk}| \geq \frac{9}{10}\cdot |Consistent|,$ where $accept_{sk}$ is the set of all secret keys $sk' \neq sk$ such that at least $\frac{1}{8}$ fraction of the message space satisfies the following: $\Pr[Verify(\pkgen(sk'), m, \sign(sk,m)] > \frac{1}{8}$. The exact constants in this and the next definition are not important.
\item \emph{``Stingy'' secret keys}:  Informally, a secret key $sk \in Consistent$ is ``stingy'' if it does not accept most signatures generated by most secret keys in $Consistent$.
More precisely, $sk \in Consistent$ is stingy iff the following is true:  $|friends_{sk}| \leq \frac{1}{2} \cdot |Consistent| $, where $friends_{sk}$ is the set of all secret keys $sk' \neq sk$ such that $sk \in accept_{sk'}$. 
\end{itemize}

The two sets above can be defined analogously with respect to a set of secret keys $S$, which is not necessarily the set $Consistent$. In this case, we denote them as $GoodSigner_S$ and $Stingy_S$.

\noindent We are now ready to define the following sequence of nested subsets of $Consistent$. Let $S_0 = Consistent$.
For $j \in [T]$, where $T$ is a large enough polynomial in $\lambda$, define 
     $$ S_j = GoodSigner_{S_{j-1}} \cap Stingy_{S_{j-1}} \,.$$


\noindent \textbf{Observation 1:} The sets $S_j$ shrink in size very quickly. More precisely, one can show that $|S_j| \leq \frac{9}{10} |S_{j-1}|$ (or $|S_j| = 1$). This is a somewhat straightforward combinatorial argument based on the fact that the $GoodSigner$ and $Stingy$ sets impose conflicting restrictions on their elements.

\vspace{2mm}
\noindent \textbf{Observation 2:}
The second crucial observation is that $sk^* \in GoodSigner_{S_j}$ for all $j$. This is for the following reason. By definition, the $S_j$ are all subsets of $Consistent$. Moreover, provided the set of queried message-signature pairs $(m_i, \sigma_i)$ is a sufficiently large polynomial, then, with overwhelming probability over the queried pairs (by a concentration bound), \emph{all} secret keys $sk \in Consistent$ accept most signatures generated using the true secret key $sk^*$ (we argued this earlier too in Equation \eqref{eq:reverse good signer}). In other words, with overwhelming probability, \emph{all} $sk \in Consistent$ belong to $accept_{sk^*}$. We emphasize that this part of the proof crucially relies on the fact that the QDS adversary can make \emph{polynomially} many signing queries (this is why our black-box separation does not contradict the one-query secure black-box construction in \cite{morimae2022quantum}).

\vspace{1mm}
Observation 2 implies that, for all $j$, $sk^* \in S_j$ if and only if $sk^* \in Stingy_{S_{j-1}}$. There are two cases:
\begin{enumerate}
\item $sk^* \in Stingy_{S_j}$ for all $j$. Then it must be that $S_T = \{sk^*\}$ (because of the shrinking property of the $S_j$). 
\item $sk^* \notin Stingy_{S_j}$ for some $j$. Then, note that, by definition of $Stingy_{S_j}$, this implies that $\pkgen(sk^*)$ accepts signatures generated by a constant fraction of $sk$ in $S_j$, with a constant probability (for a constant fraction of messages). More formally, $sk^* \in accept_{sk}$ for a constant fraction of $sk$.
\end{enumerate}

Returning to our adversary for the QDS scheme, $\mathcal{A}$ asks $\mathcal{Q}$ to do the following: compute the sets $S_j$ as defined above, and output \emph{one} uniformly random element from each $S_j$ (note that this is an exponential-time computation, and so it can be ``run'' by $\mathcal{Q}$). $\mathcal{A}$ picks a uniformly random secret key from this (polynomial-size) set, and uses it to sign a uniformly random message. By the properties we proved above, the list of secret keys output by $\mathcal{Q}$ contains $sk^*$ (in case 1), or contains (in case 2), with constant probability, an $sk$ such that $sk^* \in accept_{sk}$. Therefore, $\mathcal{A}$ wins the unforgeability game with inverse-polynomial probability.

\section{Preliminaries}
\subsection{Basic Notation}
Throughout the paper, $[n]$ denotes the set of integers $\{1,2,\dots, n\}$. If $X$ is a probability distribution, we use $x \sim X$ to denote that $x$ is sampled according to $X$. A function $f$ is \emph{negligible} if for every constant $c>0$, $f(n)\leq \frac{1}{n^c}$ for all sufficiently large $n$. 
We use the abbreviation QPT for a quantum polynomial time algorithm. We use the notation $A^{(\cdot)}$ to refer to an algorithm (classical or quantum) that makes queries to an oracle.

\subsection{Quantum Information}
We use $TD(\rho,\sigma)$ to denote the trace distance between density matrices $\rho$ and $\sigma$. For a quantum channel $\mathcal{A}$, we let $\|A\|_{\diamond}$ denote its \emph{diamond norm}. The diamond norm and trace distance satisfy the following relation:
\begin{fact}\label{fact}\cite{NC}
     Let $\mathcal{A}$ and $\mathcal{B}$ be quantum channels and $\rho$ be a density matrix. Then,
     \[
     TD(\mathcal{A}(\rho),\mathcal{B}(\rho)) \leq ||\mathcal{A}-\mathcal{B}||_{\diamond}
     \]
 \end{fact}
\subsection{Haar measure and its concentration}
We use $\mathbb{U}(N)$ to denote the group of $N \times N$ unitary matrices, and $\mu_N$ to denote the Haar measure on $\mathbb{U}(N)$. 
Given a metric space $(\mathcal{M},d)$ where $d$ denotes the metric on the set $\mathcal{M}$, a function $f: \mathcal{M}\rightarrow \mathbb{R}$ is $\mathcal{L}$-lipschitz if for all $x,y \in \mathcal{M},\left|f(x)-f(y)\right|\leq \mathcal{L}\cdot d(x,y)$. The following inequality involving Lipschitz continuous functions captures the strong concentration of Haar measure.
\begin{theorem}[\cite{Mec19}]
\label{thm:conc}
 Given $N_1,N_2, \dots, N_k \in \mathbb{N}$, let $X = \mathbb{U}(N_1)\bigoplus \dots \bigoplus  \mathbb{U}(N_k) $ be the space of block diagonal unitary matrices with blocks of size $N_1,N_2,\dots ,N_k$. Let $\nu= \nu_1 \times \dots \times \nu_k$ be the product of Haar measures on $X$. Suppose that $f: X\rightarrow \mathbb{R}$ is $\mathcal{L}$-Lipshitz with respect to the Frobenius norm. Then for every $t>0$,
 \[
 \Pr_{U \leftarrow \nu}[f(U) \geq \mathbb{E}_{V \leftarrow \nu}[f(V)]+t]\leq \exp({-\frac{(N-2)t^2}{24L^2}}) \,,
 \]
 where $N=\min{N_1,\dots ,N_k}$.
 \end{theorem}

 \begin{lemma}[\cite{kretschmer2021quantum}]
 \label{lem:lipshitz}
     Let $A^{(\cdot)}$ be a quantum algorithm that makes $T$ queries to an oracle, and let $\ket{\psi}$ be any input to $A^{(\cdot)}$. Define $f(U)=\Pr[A^{U}(\ket{\psi})=1]$. Then, $f$ is $2T$-Lipshitz with respect to the Frobenius norm.
 \end{lemma}

\subsection{Quantum pseudorandomness}

The notion of pseudorandom quantum states was introduced in \cite{ji2018pseudorandom}. The following is a formal definition. We use $\sigma_d$ to denote the Haar measure on $d$-dimensional pure quantum states.
\begin{definition}[Pseudorandom Quantum State (PRS)]\label{def: prs}
A Pseudorandom Quantum State (PRS) is a pair of $QPT$ algorithms $(\mathsf{GenKey},\mathsf{GenState})$ such that the following holds. There exists $n:\mathbb{N} \rightarrow \mathbb{N}$, and a family $\{\mathcal{K}_{\lambda}\}_{\lambda \in \mathbb{N}}$ of subsets of $\{0,1\}^*$ such that:
	   \begin{itemize}
			\item $\mathsf{GenKey}(1^{\lambda}) \rightarrow k$: Takes as input a security parameter $\lambda$, and outputs a key $k \in \mathcal{K}_{\lambda}$.
			\item $\mathsf{GenState}(k) \rightarrow \ket{\PRS(k)}$: Takes as input a key $k \in \mathcal{K}_{\lambda}$, for some $\lambda$, and outputs an $n(\lambda)$-qubit state. We additionally require that the state on input $k$ be unique, and we denote this as $\ket{\PRS(k)}$.
        \end{itemize}
Moreover, the following holds. For any (non-uniform) QPT quantum algorithm $A$, and any $m = poly$, there exists a negligible function $\negl$ such that, for all $\lambda \in \mathbb{N}$,
			$$
			\abs{\Pr_{k \gets \mathsf{GenKey}(1^{\lambda})}[A\big( \ket{\PRS(k)}^{\otimes m(\lambda)} \big) = 1] -
				\Pr_{\ket{\psi} \gets \sigma_{2^{n(\lambda)}}}[A\big(\ket{\psi}^{\otimes m(\lambda)} \big) = 1]} \leq \negl(n) \,.$$
  \end{definition}

\begin{definition}[Pseudorandom unitary transformations (PRU) (\cite{JLS2018})]
Let $n: \mathbb{N} \rightarrow \mathbb{N}$. Let $\{\mathcal{U}_{\lambda}\}_{\lambda \in \mathbb{N}}$ be a family of unitaries where $\mathcal{U}_{\lambda}$ is a family of $n(\lambda)$-qubit unitaries $\{U_k\}_{k \in \{0,1\}^{\lambda}}$. We say that $\{\mathcal{U}_{\lambda}\}_{\lambda \in \mathbb{N}}$ is pseudorandom if the following conditions hold:
\begin{enumerate}
    \item (Efficient computation) There is a QPT algorithm $G$ that implements
$U_k$ on input $k$, meaning that for any $n$-qubit input $\ket{\psi}$, $G(k,\ket{\psi})= U_k\ket{\psi}$.
    \item (Computationally indistinguishable) For any QPT algorithm $A^{(\cdot)}$, there exists a negligible function $\negl$ such that, for all $\lambda$,
    \[
    \left|\Pr_{k \gets \{0,1\}^{\lambda}}[A^{U_k}(1^{\lambda})=1]-\Pr_{U \leftarrow \mu_{2^n}}[A^{U}(1^{\lambda})=1]\right|\leq \negl(\lambda) \,.
    \]
\end{enumerate}
    
\end{definition}

\section{On quantum oracle separations and black-box constructions}
\label{sec:black-box-constructions}
In this section, we clarify what we mean by a ``black-box construction'' of primitive $\mathcal{Q}$ from primitive $\mathcal{P}$ when the primitives involve \emph{quantum} algorithms (and possibly quantum state outputs). We also clarify the relationship between a \emph{quantum} oracle separation of $\mathcal{P}$ and $\mathcal{Q}$ and the (im)possibility of a black-box construction of one from the other. To the best of our knowledge, while black-box separations in the quantum setting have been the topic of several recent works, a somewhat formal treatment of the terminology and basic framework is missing. This section also appears verbatim in the concurrent work \cite{chen2024power}.

In the quantum setting, it is not immediately obvious what the correct notion of ``black-box access'' is. There are a few reasonable notions of what it means for a construction to have ``black-box access'' to another primitive. We focus on three variants: \emph{unitary} access, \emph{isometry} access, and access to \emph{both the unitary and its inverse}.

The summary is that, similarly to the classical setting, a \emph{quantum} oracle separation of primitives $\mathcal{P}$ and $\mathcal{Q}$ (i.e.\ a quantum oracle relative to which $\mathcal{P}$ exists but $\mathcal{Q}$ does not) implies the impossibility of a black-box construction of $\mathcal{Q}$ from $\mathcal{P}$, but with one caveat: the type of oracle separation corresponds directly to the type of black-box construction that is being ruled out. For example,  the oracle separation needs to be ``closed under giving access to the inverse of the oracle'', i.e.\ the separation needs to hold relative to an oracle \emph{and} its inverse. We start by introducing some terminology.

\paragraph{Terminology.} 
A quantum channel is a CPTP (completely-positive-trace-preserving) map. The set of quantum channels captures all admissible ``physical'' processes in quantum information, and it can be thought of as the quantum analogue of the set of functions  $f: \{0,1\}^* \rightarrow \{0,1\}^*$. 

For the purpose of this section, a quantum channel is specified by a family of unitaries $\{U_n\}_{n \in \mathbb{N}}$ (where $U_n$ acts on an input register of size $n$, and a work register of some size $s(n)$). The quantum channel maps an input (mixed) state $\rho$ on $n$ qubits to the (mixed) state obtained as follows: apply $U_n (\cdot) U_n^{\dagger}$ to $\rho \otimes (\ket{0}\bra{0})^{\otimes s(n)}$; measure a subset of the qubits; output a subset of the qubits (measured or unmeasured). We say that the family $\{U_n\}_{n \in \mathbb{N}}$ is a \emph{unitary implementation} of the quantum channel. We say that the quantum channel is QPT if it possesses a unitary implementation $\{U_n\}_{n \in \mathbb{N}}$ that is additionally a uniform family of efficiently computable unitaries. In other words, the quantum channel is implemented by a QPT algorithm.

One can also consider the family of isometries $\{V_n\}_{n \in \mathbb{N}}$ where $V_n$ takes as input $n$ qubits, and acts like $U_n$, but with the work register fixed to $\ket{0}^{s(n)}$, i.e.\
$V_n: \ket{\psi} \mapsto U_n(\ket{\psi}\ket{0}^{\otimes s(n)})$. We refer to $\{V_n\}_{n \in \mathbb{N}}$ as the \emph{isometry implementation} of the quantum channel. 

We will also consider QPT algorithms with access to some oracle $\mathcal{O}$. In this case, the unitary (resp. isometry) implementation $\{U_n\}_{n \in \mathbb{N}}$ should be \emph{efficiently computable given access to $\mathcal{O}$}.

Before diving into formal definitions, a bit informally, a \emph{primitive} $\mathcal{P}$ can be thought of as a set of conditions on tuples of algorithms $(G_1, \ldots, G_k)$. 
For example, for a digital signature scheme, a valid tuple of algorithms is a tuple $(\textit{Gen}, \textit{Sign}, \textit{Verify})$ that satisfies ``correctness'' (honestly generated signatures are accepted by the verification procedure with overwhelming probability) and ``security'' (formalized via an unforgeability game). Equivalently, one can think of the tuple of algorithms $(G_1, \ldots, G_k)$ as a \emph{single} algorithm $G$ (with an additional control input). 

A thorough treatment of black-box constructions and reductions in the classical setting can be found in \cite{reingold2004notions}. Our definitions are a quantum analog of those in \cite{reingold2004notions}. They follow the latter style whenever possible, and they depart from it whenever necessary.

\begin{definition}
A \emph{primitive} $\mathcal{P}$ is a pair $\mathcal{P} = (\mathcal{F}_{\mathcal{P}}, \mathcal{R}_{\mathcal{P}})$\footnote{Here $\mathcal{F}_{\mathcal{P}}$ should be thought of as capturing the ``correctness'' property of the primitive, while $\mathcal{R}_{\mathcal{P}}$ captures ``security''.} where $\mathcal{F}_{\mathcal{P}}$ is a set of quantum channels, and $\mathcal{R}_{\mathcal{P}}$ is a relation over pairs $(G, A)$ of quantum channels, where $G \in \mathcal{F}_{\mathcal{P}}$.

A quantum channel $G$ is an \emph{implementation} of $\mathcal{P}$ if $G \in \mathcal{F}_{\mathcal{P}}$. If $G$ is additionally a QPT channel, then we say that $G$ is an \emph{efficient implementation} of $\mathcal{P}$ (in this case, we refer to $G$ interchangeably as a QPT channel or a QPT algorithm). 

A quantum channel $A$ (usually referred to as the ``adversary'') $\mathcal{P}$-breaks $G \in \mathcal{F}_{\mathcal{P}}$ if $(G, A) \in \mathcal{R}_{\mathcal{P}}$. We say that $G$ is a \emph{secure implementation} of $\mathcal{P}$ if $G$ is an implementation of $\mathcal{P}$ such that no QPT channel $\mathcal{P}$-breaks it. The primitive $\mathcal{P}$ \emph{exists} if there exists an efficient and secure implementation of $\mathcal{P}$.

Let $U$ be a unitary (resp.\ isometry) implementation of $G \in \mathcal{P}$. Then, we say that $U$ is a \emph{unitary (resp. isometry) implementation} of $\mathcal{P}$. For ease of exposition, we also say that quantum channel $A$ $\mathcal{P}$-breaks $U$ to mean that $A$ $\mathcal{P}$-breaks $G$.

\end{definition}
Since we will discuss oracle separations, we give corresponding definitions \emph{relative to an oracle}. Going forward, for ease of exposition, we often identify a quantum channel with the algorithm that implements it.
\begin{definition}[Implementations relative to an oracle]
\label{def:oracle-implementation}
Let $\mathcal{O}$ be a unitary (resp. isometry) oracle. An \emph{implementation} of primitive $\mathcal{P}$ relative to $\mathcal{O}$ is an oracle algorithm $G^{(\cdot)}$ such that $G^{\mathcal{O}} \in \mathcal{P}$\footnote{We clarify that here $G^{\mathcal{O}}$ is only allowed to query the unitary $\mathcal{O}$, not its inverse. However, as will be the case later in the section, $\mathcal{O}$ itself could be of the form $\mathcal{O} = (W, W^{-1})$ for some unitary $W$.}. We say the implementation is efficient if $G^{(\cdot)}$ is a QPT oracle algorithm.

Let $U$ be a unitary (resp.\ isometry) implementation of $G^{\mathcal{O}}$. Then, we say that $U$ is a \emph{unitary (resp.\ isometry) implementation} of $\mathcal{P}$ \emph{relative to $\mathcal{O}$}.
\end{definition}


\begin{definition}
\label{def:exist-relative-to-oracle}
    We say that a primitive $\mathcal{P}$ exists relative to an oracle $\mathcal{O}$ if: 
\begin{itemize}
	\item[(i)] There exists an efficient implementation $G^{(\cdot)}$ of $\mathcal{P}$ relative to $\mathcal{O}$, i.e.\ $G^{\mathcal{O}} \in \mathcal{P}$ (as in Definition~\ref{def:oracle-implementation}).
	\item[(ii)] The security of $G^{\mathcal{O}}$ holds against all QPT adversaries that have access to $\mathcal{O}$. More precisely, for all QPT $A^{(\cdot)}$, $(G^{\mathcal O},A^{\mathcal O})\notin \mathcal{R}_{\mathcal{P}}$.
\end{itemize} 
\end{definition}

\vspace{3mm}
There are various notions of black-box constructions and reductions (see, for example, \cite{reingold2004notions}). Here, we focus on (the quantum analog of) the notion of a \emph{fully black-box construction}. We identify and define three analogs based on the type of black-box access available to the construction and the security reduction.

\begin{definition}
\label{def:bb-unitary}
A QPT algorithm $G^{(\cdot)}$ is a \emph{fully black-box construction of $\mathcal{Q}$ from \textbf{unitary access} to $\mathcal{P}$} if the following two conditions hold:
\begin{enumerate}
	\item (\emph{black-box construction with unitary access}) For every unitary implementation $U$ of $\mathcal{P}$, $G^{U}$ is an implementation of $\mathcal{Q}$.
	\item (\emph{black-box security reduction with unitary access}) There is a QPT algorithm $S^{(\cdot)}$ such that, for every unitary implementation $U$ of $\mathcal{P}$, every adversary $A$ that $\mathcal{Q}$-breaks $G^U$, and every unitary implementation $\tilde{A}$ of $A$, it holds that $S^{\tilde{A}}$ $\mathcal P$-breaks $U$.
\end{enumerate}
\end{definition}

\begin{definition}
\label{def:bb-isometry}
	A QPT algorithm $G^{(\cdot)}$ is a \emph{fully black-box construction of $\mathcal{Q}$ from \textbf{isometry access} to $\mathcal{P}$} if the following two conditions hold:
	\begin{enumerate}
	\item (\emph{black-box construction with isometry access}) For every isometry implementation $V$ of $\mathcal{P}$, $G^{V}$ is an implementation of $\mathcal{Q}$.
	\item (\emph{black-box security reduction with isometry access}) There is a QPT algorithm $S^{(\cdot)}$ such that, for every isometry implementation $V$ of $\mathcal{P}$, every adversary $A$ that $\mathcal{Q}$-breaks $G^V$, and every isometry implementation $\tilde{A}$ of $A$, it holds that $S^{\tilde{A}}$ $\mathcal P$-breaks $V$.
\end{enumerate}
\end{definition}

\begin{definition}
\label{def:bb-unitary-and-inverse}
	A QPT algorithm $G^{(\cdot)}$ is a \emph{fully black-box construction of $\mathcal{Q}$ from $\mathcal{P}$ \textbf{with access to the inverse}} if the following two conditions hold:
	\begin{enumerate}
		\item (\emph{black-box construction with access to the inverse}) For every unitary implementation $U$ of $\mathcal{P}$, $G^{U, U^{-1}}$ is an implementation of $\mathcal{Q}$.
		\item (\emph{black-box security reduction with access to the inverse}) There is a QPT algorithm $S^{(\cdot)}$ such that, for every unitary implementation $U$ of $\mathcal{P}$, every adversary $A$ that $\mathcal{Q}$-breaks $G^{U, U^{-1}}$, and every unitary implementation $\tilde{A}$ of $A$, it holds that $S^{\tilde{A}, \tilde{A}^{-1}}$ $\mathcal P$-breaks $U$\footnote{One could define even more variants of "fully black-box constructions" by separating the type of access that $G$ has to the implementation of $\mathcal{P}$ from the type of access that $S$ has to $A$ (currently they are consistent in each of Definitions \ref{def:bb-unitary}, \ref{def:bb-isometry}, and \ref{def:bb-unitary-and-inverse}). Here, we choose to limit ourselves to the these three definitions.}.
	\end{enumerate}
\end{definition}

We now clarify the relationship between a \emph{quantum} oracle separation of primitives $\mathcal{P}$ and $\mathcal{Q}$ and the (im)possibility of a black-box construction of one from the other. 

The following is a quantum analog of a result by Impagliazzo and Rudich~\cite{impagliazzo1989limits} (formalized in \cite{reingold2004notions} using the above terminology).
\begin{theorem}
	\label{thm:quantumIR}
	 Suppose there exists a fully black-box construction of primitive $\mathcal{Q}$ from unitary (resp.\ isometry) access to primitive $\mathcal{P}$. Then, for every unitary (resp.\ isometry) $\mathcal{O}$, if $\mathcal{P}$ exists relative to $\mathcal{O}$, then $\mathcal{Q}$ also exists relative to $\mathcal{O}$.
\end{theorem}
This implies that a unitary (resp.\ isometry) oracle separation (i.e.\ the existence of an oracle relative to which $\mathcal{P}$ exists but $\mathcal{Q}$ does not) suffices to rule out a fully black-box construction of $\mathcal{Q}$ from unitary (resp.\ isometry) access to $\mathcal{P}$.

\begin{proof}[Proof of Theorem \ref{thm:quantumIR}]
We write the proof for the case of unitary access to $\mathcal{P}$. The proof for the case of isometry access is analogous (replacing unitaries with isometries).
	Suppose there exists a fully black-box construction of $\mathcal{Q}$ from $\mathcal{P}$. Then, by definition, there exist QPT $G^{(\cdot)}$ and $S^{(\cdot)}$ such that:
	\begin{enumerate}
		\item (\emph{black-box construction}) For every unitary implementation $U$ of $\mathcal{P}$, $G^{U}$ is an implementation of $\mathcal{Q}$.
		\item (\emph{black-box security reduction}) For every implementation $U$ of $\mathcal{P}$, every adversary $A$ that $\mathcal{Q}$-breaks $G^U$, and every unitary implementation $\tilde{A}$ of $A$, it holds that $S^{\tilde{A}}$ $\mathcal P$-breaks $U$.
	\end{enumerate}
Let $\mathcal O$ be a quantum oracle relative to which $\mathcal{P}$ exists. Since, by Definition~\ref{def:exist-relative-to-oracle}, $\mathcal{P}$ has an \emph{efficient} implementation relative to $\mathcal{O}$, there exists a uniform family of unitaries $U$ that is \emph{efficiently computable} with access to $\mathcal{O}$, such that $U$ is a unitary implementation of $\mathcal{P}$. Moreover, $U$ (or rather the quantum channel that $U$ implements) is a secure implementation of $\mathcal{P}$ relative to $\mathcal{O}$.

We show that the following QPT oracle algorithm $\tilde{G}^{(\cdot)}$ is an efficient implementation of $\mathcal{Q}$ relative to $\mathcal{O}$, i.e.\   $\tilde{G}^{\mathcal O} \in \mathcal{Q}$. $\tilde{G}^{\mathcal O}$ runs as follows: implement $G^{U}$ by running $G$, and simulate each call to $U$ by making queries to $\mathcal O$. Note that $\tilde{G}^{(\cdot)}$ is QPT because $U$ is a uniform family of efficiently computable unitaries given access to $\mathcal{O}$. Since $\tilde{G}^{\mathcal O}$ is equivalent to $G^{U}$, and $G^U \in \mathcal{Q}$ (by property 1 above), then $\tilde{G}^{\mathcal O} \in \mathcal{Q}$.

We are left with showing that $\tilde{G}^{\mathcal O}$ is a secure implementation relative to $\mathcal O$, i.e.\ that there is no QPT adversary $A^{(\cdot)}$ such that  $A^{\mathcal O}$ $\mathcal{Q}$-breaks $\tilde{G}^{\mathcal O}$. Suppose for a contradiction that there was a QPT adversary $A^{(\cdot)}$ such that $\mathcal{A}^{\mathcal{O}}$ $\mathcal{Q}$-breaks $\tilde{G}^{\mathcal O}$ (which is equivalent to $G^{U}$). Then, by property 2, $S^{A^\mathcal O}$ $\mathcal{P}$-breaks $U$. Note that adversary $S^{A^{\mathcal{O}}}$ can be implemented efficiently with oracle access to $\mathcal O$, because both $S^{(\cdot)}$ and $A^{(\cdot)}$ are QPT. Thus, this contradicts the security of $U$ relative to $\mathcal{O}$ (formally, of the quantum channel that $U$ implements).
\end{proof}

Similarly, we state a version of Theorem \ref{thm:quantumIR} for fully black-box constructions with access to the inverse.
\begin{theorem}
	\label{thm:quantumIRinverse}
	Suppose there exists a fully black-box construction of primitive $\mathcal{Q}$ from primitive $\mathcal{P}$ with access to the inverse.  Then, for every unitary $\mathcal O$, if $\mathcal{P}$ exists relative to $(\mathcal O, \mathcal O^{-1})$, then $\mathcal{Q}$ also exists relative to the oracle $(\mathcal O, \mathcal O^{-1})$. 
\end{theorem}
\begin{proof}
The proof is analogous to the proof of Theorem \ref{thm:quantumIR}. The only difference is that now $G^{(\cdot)}$ additionally makes queries to the inverse of the unitary implementation $U$ of $\mathcal{P}$. Since $U^{-1}$ can be implemented efficiently given access to $(\mathcal O, \mathcal O^{-1})$, we can now define an efficient implementation $\tilde{G}^{(\cdot)}$ of $\mathcal{P}$ relative to $(\mathcal O, \mathcal O^{-1})$. Proving that $\tilde{G}^{\mathcal{O}, \mathcal{O}^{-1}}$ is a secure implementation of $\mathcal{P}$ relative to $(\mathcal O, \mathcal O^{-1})$ also proceeds analogously.
\end{proof}

\section{Quantum Public Key Digital Signatures}
In this section, we define QDS schemes with \emph{quantum} public keys (and classical secret key and signatures). Going forward, unless we specify otherwise, we use the term QDS to refer to this type of signature scheme.
\begin{definition}
\label{def:qds}
A Quantum Digital Signature scheme (QDS) is a tuple of algorithms $(\skgen, \pkgen, \sign, \verify)$ satisfying the following:
\begin{itemize}
    \item $\skgen(1^{\lambda}) \rightarrow sk$: is a QPT algorithm that takes as input  $1^{\lambda}$, and outputs a classical secret key $sk$. We assume that $sk$ has length $\lambda$. 
    \item $\pkgen(\textit{sk}) \rightarrow \ket{pk}$: is a deterministic algorithm that takes as input a secret key $sk$, and outputs the quantum state $\ket{pk}$.\footnote{To clarify, $\ket{pk}$ is allowed to be an arbitrary pure state (not necessarily a standard basis state).} We additionally require $\ket{pk}$ to be fixed given $sk$, i.e.\ the algorithm $\skgen$ is deterministic (it consists of a fixed unitary quantum circuit acting on the input $sk$, and some auxiliary registers).\footnote{We include this requirement so that the notion of ``quantum'' public key is a little more faithful to the spirit of a classical public key. This requirement ensures that the party in possession of the secret key can generate multiple copies of the corresponding public key. Note that for a completely classical digital signature scheme this requirement is without loss of generality, since any randomness used in the generation procedure can be included in the secret key.}
      
    \item $\sign(sk,m) \rightarrow \sigma$: is a QPT algorithm that takes as input a secret key $sk$ and a classical message $m$ from some message space (that may depend on the security parameter), and outputs a classical signature $\sigma$. For a security parameter $\lambda$, we denote by $\mathcal{M}_{\lambda}$ the corresponding message space.

    \item $\verify(\ket{pk},m,\sigma) \rightarrow \mathsf{accept}/\mathsf{reject}$: is a QPT algorithm that takes as input a public key $\ket{pk}$, a message $m$, and a candidate signature $\sigma$, and outputs $\mathsf{accept}$ or $\mathsf{reject}$.
\end{itemize}
We require the following ``correctness'' property. There exists a negligible function $\negl$ such that, for all $\lambda \in \mathbb{N}$, the following holds except with probability $\negl(\lambda)$ over sampling $sk \gets \textit{SKGen}(1^{\lambda})$. For all $m \in \mathcal{M}_{\lambda}$, 
$$ \Pr[Ver(\textit{PKGen}(sk), Sign(sk, m)) = 1] \geq 1 - \negl(\lambda) \,. $$

If there is a function $\ell: \mathbb{N}\rightarrow \mathbb{N}$, such that, for all $\lambda$, $\mathcal{M}_\lambda$ is the set of strings of length $\ell(\lambda)$, then we say that the scheme is a QDS for messages of length $\ell(\lambda)$.
\end{definition}

For simplicity, we will consider QDS where each $\mathcal{M}_{\lambda}$ is the set of strings of a certain length. 

\paragraph{Multi-time security}
The notion of security that we focus on in this work is the standard ``multi-time'' security. The latter allows the adversary to make an arbitrary polynomial number of queries to a signing oracle, before having to produce a valid signature of an ``unqueried'' message. Formally, multi-time security is defined in terms of the following security game between a challenger $\mathcal{C}$ and an adversary $\mathcal{A}$.

\begin{enumerate}
\item $\mathcal{C}$ runs $sk \leftarrow \skgen(1^{\lambda})$.
\item $\mathcal{A}$ receives $\ket{pk}^{\otimes t(\lambda)}$ (for some polynomially-bounded function $t$ that depends on $\mathcal{A}$).
    \item For each $i \in \{1, \dots t\}$, $\mathcal{A}$ sends a message $m_i$ to $\mathcal{C}$;  $\mathcal{C}$ runs $\sigma_i \leftarrow \sign(m_i,sk)$, and sends $\sigma_i$ to $\mathcal{A}$.
    \item $\mathcal{A}$ sends $(m,\sigma)$ to $\mathcal{C}$ such that $m \notin \{m_1, \dots m_t\}$.
    \item $\mathcal{C}$ outputs 1 iff $\verify(\ket{pk},m,\sigma)$ accepts.
\end{enumerate}
Let $\textsf{multi-time}(\lambda,\mathcal{A})$ be a random variable denoting the output of the game above. 

\begin{definition}[multi-time security]
\label{def:polytimesec}
A QDS scheme satisfies multi-time security if, for all QPT adversaries $\mathcal{A}$, there exists a negligible function $\negl$ such that
\begin{equation}
\Pr[\textsf{multi-time}(\lambda, \mathcal{A}) = 1] = \negl(\lambda) \,.
\end{equation}
\end{definition}

In this work, we consider QDS schemes relative to some oracle $O$. In this setting, the syntax and security definitions are identical to the ones given in this section, except that all of the algorithms, including the adversary, have access to $O$.



\section{Oracle separation of Quantum Digital Signatures and PRS}
\label{sec:sim-haar}
In this section, we describe an oracle relative to which PRS exist, and a QDS scheme that satisfies multi-time security (Definitions~\ref{def:qds} and \ref{def:polytimesec}) does not. Most of this section is dedicated to proving the latter. For the rest of the section, unless we mention otherwise, any QDS scheme that we consider has a quantum public key and classical signatures. We also restrict our attention to multi-time security.

\paragraph{The separating oracle.}
The oracle is similar to the one used by Kretschmer in \cite{kretschmer2021quantum}. The oracle $O$ consists of a pair of oracles $(\mathcal{U}, \mathcal{Q})$, where $\mathcal{Q}$ is a classical oracle solving a fixed $\mathsf{EXP}$-complete problem, and $\mathcal{U}$ is a collection of Haar-random unitaries $\{\mathcal{U}_{\ell}\}_{\ell \in \mathbb{N}}$, where each $\mathcal{U_{\ell}}$ is an indexed list of $2^{\ell}$ Haar-random unitaries acting on $\ell$ qubits.

\vspace{2mm} 
In this section, we show the following two theorems.

\begin{theorem}
\label{thm:prs-existence}
    With probability 1 over $\mathcal{U}$, there exists a family of PRUs relative to $\mathcal{(U,Q)}$.
\end{theorem}

\begin{theorem}
\label{thm:no-qds}
    Let $\ell: \mathbb{N} \rightarrow \mathbb{N}$ be such that $\ell(\lambda) \geq 2 \cdot \log(\lambda)$ for all sufficiently large $\lambda$. Then, relative to $\mathcal{(U,Q)}$, there does not exist a QDS scheme for messages of length $\ell(\lambda)$ (where $\lambda$ is the security parameter).
\end{theorem}

As mentioned previously, we can also show that, relative to $\mathcal{(U,Q)}$, QDS schemes for shorter message length do not exist under a slightly different unforgeability definition than Definition \ref{def:polytimesec}. In this definition, for the adversary to succeed, it suffices to produce a a valid signature that has not been previously produced by the challenger (even if this is the signature of a previously queried message). We discuss this in more detail in Remark \ref{rem:technicality1}.

\subsection{Existence of PRS relative to the oracle}

Theorem \ref{thm:prs-existence} follows from the fact that an appropriate family of pseudorandom unitaries (PRU) exists relative to $O$. The latter was shown by Kretschmer\footnote{Kretschmer shows this for a slightly different oracle $(\mathcal{U}, \mathcal{Q})$, where $\mathcal{U}$ is the same, but $\mathcal{Q}$ is a PSPACE oracle, instead of an EXP oracle. The proof is analogous. The only step in the proof where this comes into play is in Lemma 31 from \cite{kretschmer2021quantum} which essentially provides a lower bound on the number of queries to $\mathcal{U}$. Crucially this lower bound holds against any unbounded quantum algorithm.} \cite{kretschmer2021quantum}.


The construction of the PRU family is very natural. For input length $\lambda$, the family is precisely $\mathcal{U_{\lambda}}$. This is a PRU family on $n(\lambda) = \lambda$ qubits. For the proof, we refer the reader to Theorem 32 in \cite{kretschmer2021quantum}. For any $n(\lambda) \geq \lambda$, we can construct a corresponding PRU family $\{\mathcal{U}'_{\lambda}\}$ on $n(\lambda)$ qubits by taking $\mathcal{U}'_{\lambda} = \mathcal{U}_{n(\lambda)}$. Security follows analogously.  

Now, a PRU family $\{\mathcal{U}_{\lambda}\}_{\lambda \in \mathbb{N}}$ on $n(\lambda)$ qubits immediately implies a PRS with output length $n(\lambda)$ as follows. Let $\mathcal{U}_{\lambda} = \{U_k\}_{k \in \{0,1\}^{\lambda}}$. Then, for $k\in \{0,1\}^{\lambda}$, one defines $\ket{\PRS(k)} = U_k \ket{0}$.

The rest of this section, is dedicated to proving Theorem \ref{thm:no-qds}.

\subsection{Simulating Haar random unitaries}
Before describing an adversary that, with query-bounded access to $(\mathcal{U, O})$, breaks the security of any QDS scheme, we need to introduce a crucial procedure that allows an adversary to ``simulate'' certain kinds of interactions with the oracle $\mathcal{U}$, by only making a small number of queries to $\mathcal{U}$. This simulation procedure was introduced by Kretschmer \cite{kretschmer2021quantum}. We give a high-level description first.

\subsubsection{Kretschmer's simulation procedure: a high-level description}
The key property of the Haar measure that makes the following simulation possible is its strong concentration, which we describe informally. Let $C$ be a fixed binary-output quantum circuit that makes $T$ queries to a Haar random unitary on $n$ qubits (i.e.\ on a Hilbert space of dimension $2^n$). Then, with high probability over sampling a pair of Haar random unitaries $U, U'$, the output distributions of the circuits $C^U$ and $C^{U'}$ (say on the $\ket{0}$ input) are within small TV distance. Quantitatively, for any $\delta>0$,
\begin{equation}
\label{eq:concentration informal}
\Pr_{U, U' \leftarrow \mu_{2^n}}\left[\left| \Pr[C^U(\ket{0}) = 1] -  \Pr[C^{U'}(\ket{0}) = 1] \right| \leq \delta \right] \leq \exp\left(-\Omega\left(\frac{2^n \cdot \delta^2}{T^2}\right)\right) \,.
\end{equation}
For example, when $T = poly(n)$, and $\delta$ is as small $2^{-c\cdot n}$ for $c<\frac12$, the upper bound is \emph{doubly} exponentially small in $n$. Thus, the concentration is strong enough to support a union bound over \emph{all} standard basis inputs. So, with overwhelming probability over $U$ and $U'$, the output distributions of $C^U$ and $C^{U'}$ on \emph{all} standard basis inputs are within a small constant TV distance of each other. 

How does this help an adversary for a QDS scheme? Suppose the adversary is trying to perform a brute-force search over inputs to a circuit $C$ (for example the $Verify$ circuit of a QDS scheme, in the hope of finding an accepting input). This search would normally require exponentially many queries to $\mathcal{U}$ (since there are exponentially many possible inputs to the circuit). The concentration property tells us that we can ignore the particular oracle $\mathcal{U}$, and instead replace queries to $\mathcal{U}$ with queries to a family of freshly sampled Haar unitaries (or $T$-designs) by paying only a small cost in TV distance.

There is only one issue. Recall that $\mathcal{U}$ is a \emph{family} of unitaries $\{\mathcal{U}_{\ell}\}_{\ell \in \mathbb{N}}$, where each $\mathcal{U_{\ell}}$ is an indexed list of $2^{\ell}$ Haar-random unitaries acting on $\ell$ qubits. When the dimension is small, relative to the number of queries, e.g.\ the number of qubits is $\ell = o(\log T)$, the concentration property no longer holds (or rather the upper bound in \eqref{eq:concentration informal} becomes trivial)! Fortunately, this issue can be circumvented because, for small enough dimension, the unitaries can be ``learnt'' to high precision efficiently by performing process tomography. More precisely, process tomography allows one to learn an arbitrary unitary on a space of dimension $d$ by making only $poly(d)$ queries. Thus, let $\lambda$ be a security parameter. Then, unitaries in $\mathcal{U}_{\ell}$, for $\ell = O(\log \lambda)$, can be learnt with only $poly(\lambda)$ queries. Moreover, notice that when $\ell = O(\log \lambda)$, there are only a total of $2^{\ell} = poly(\lambda)$ unitaries in $\mathcal{U}_\ell$, and so one can learn \emph{all} of them with a total of only $poly(\lambda)$ time and queries. On the other hand, unitaries in $\mathcal{U}_\ell$ for $\ell \geq c \cdot \log \lambda$, for a large enough constant $c$, enjoy a strong enough concentration, and can thus be replaced with fresh Haar unitaries, or $T$-designs, at a small cost in TV distance.

To summarize, let $C$ be a quantum circuit on inputs of size $\lambda$. Suppose $C$ makes $poly(\lambda)$ queries to $\mathcal{U}$. With high probability over $\mathcal{U}$, the output distribution of $C$ on \emph{all} standard basis inputs can be simulated to within a small constant in TV distance with only $poly(\lambda)$ queries to $\mathcal{U}$ by:
\begin{itemize}
    \item Using $T$-designs to replace the oracle calls to $\mathcal{U}_{\ell}$, for large enough $\ell$ (such as $\ell \geq  c \cdot \log \lambda$), for some large enough constant $c$). This perfectly simulates $T$ queries to a freshly sampled family $\mathcal{U}_{\ell}$, and, by the concentration property, this in turn approximates the original output distribution of each run to within a small $\delta$, where we can take $\delta$ to be a small constant\footnote{Note that we cannot take $\delta$ to be exponentially small in $\lambda$ here because $\ell$ could be as small as $c \cdot \log \lambda$.}. This step does not require any queries to $\mathcal{U}$.
    \item Efficiently learning all of the unitaries in $\mathcal{U}_{\ell}$, for all small enough $\ell$ (e.g.\ $\ell \leq \Theta(\log \lambda)$). This step requires $poly(\lambda)$ time and queries to $\mathcal{U}$.
\end{itemize}

\subsubsection{Kretschmer's simulation procedure: a formal description}
Before formally describing the simulation procedure, we introduce the necessary lemmas relating to process tomography and $t$-designs. For a unitary $U$, we let $U(\cdot) U^{\dag}$ denote the the quantum channel $\rho \rightarrow U \rho U^{\dag}$.
\paragraph{Process Tomography}
We focus on process tomography for a unitary channel. In this setting, the algorithm is given black-box access to a unitary $Z$. After a number of queries to $Z$, the algorithm should output a classical description of a unitary $\Tilde{Z}$ with the goal of minimizing the following quantity:
\[
\mathbb{E}\left[\left\|Z(\cdot) Z^{\dag}-\Tilde{Z}(\cdot) \Tilde{Z}^{\dag}\right\|_{\diamond}\right] \,.
\]
In \cite{HKOT23}, the authors describe an algorithm with the following guarantees. This algorithm is used as a sub-routine in Kretschmer's simulation procedure.
\begin{theorem}[\cite{HKOT23}]
\label{cor:processtom}
 There is a quantum algorithm that, given $\epsilon, \mu \in (0,1)$, as well as black-box access to an unknown $d$-dimensional unitary $Z \in \mathbb{U}(d)$, makes $O(\frac{d^2}{\epsilon}\log(1/\eta))$ queries to the black box and outputs a classical description of a unitary $\Tilde{Z} \in \mathbb{U}(d)$ that is $\epsilon$-close to $Z$ in diamond norm with probability at least $1-\eta$. Specifically, $\mathbb{E}\left[\left\| Z (\cdot) Z^{\dag}-\Tilde{Z}(\cdot) \Tilde{Z}^{\dag}\right\|_{\diamond}^2\right]\leq \epsilon^2$, which implies that $\Pr\left[\left\|Z(\cdot)Z^{\dag}-\Tilde{Z}(\cdot)\Tilde{Z}^{\dag}\right\|_{\diamond}\leq 3\epsilon\right] \geq 1-\mu$.
 The gate complexity of this algorithm is $poly(d,1/\epsilon)$, and its classical running time is $poly(d,1/\epsilon)$.
\end{theorem}
\paragraph{Approximate $T$ Designs}
An $\epsilon$-approximate quantum unitary $t$-design is a distribution over unitaries that ``$\epsilon$-approximates'' a Haar random unitary, when considering their action via a $t$-copy parallel repetition.
\begin{definition}[Approximate Unitary Design \cite{BHH}]
Let $\epsilon \in [0,1],  t \in \mathbb{N}$. A probability distribution $S$ over $\mathbb{U}(N)$ is an $\epsilon$-approximate unitary $t$-design if:
\[
(1-\epsilon)\mathbb{E}_{U \leftarrow \mu_N}[(U(.)U^{\dag})^{\otimes t}]\preceq\mathbb{E}_{U \sim S}[(U(.)U^{\dag})^{\otimes t}] \preceq (1-\epsilon)\mathbb{E}_{U \leftarrow \mu_N}[(U(.)U^{\dag})^{\otimes t}] \,,
\]
where $B \preceq A$ means that $B-A$ is positive semidefinite.   
\end{definition}
It is well-known that there are efficient constructions of such unitary $t$-designs.
\begin{lemma}[\cite{BHH}]
\label{lem:tdescons}
There exists $m: \mathbb{N} \rightarrow \mathbb{N}$, such that the following holds. For each $n, t \in \mathbb{N}$, and $\epsilon>0$, there is a $poly(n,t,\log(\frac{1}{\epsilon}))$-time classical algorithm $A$ that takes $m(n)$ bits of randomness as input, and outputs a description of a unitary quantum circuit on $n$ qubits such that the output distribution of $A$ is an $\epsilon$-approximate unitary $t$-design (over $\mathbb{U}(2^n)$).
\end{lemma}
An $\epsilon$-approximate unitary $t$-design $S$ is said to be \emph{phase-invariant} if, for any unitary $U$ in the support of $S$, $U$ and $\omega U$ are sampled with the same probability, where $\omega$ is the $(t+1)$-th root of unity. The following lemma says that replacing the Haar measure with a phase-invariant $\epsilon$-approximate unitary $t$-design is undetectable to any algorithm that makes only $t$ queries to the Haar random unitary.
\begin{lemma}[\cite{kretschmer2021quantum}]
 \label{lem:tdesign}
     Let $S$ be a phase-invariant $\epsilon$ approximate unitary $t$-design over $\mathbb{U}(N)$, and let $D^{(\cdot)}$ be any $t$-query quantum algorithm. Then, 
     \[
     (1-\epsilon)\Pr_{U\leftarrow \mu_{N}}[D^U=1] \leq \Pr_{U\leftarrow S}[D^U=1] \leq (1+\epsilon)\Pr_{U\leftarrow \mu_{N}}[D^U=1]
     \]
  (where in the above, $D^{(\cdot)}$ is also allowed to make queries to controlled-$U$).
 \end{lemma}
\vspace{2mm}
We are now ready to describe Kretschmer's simulation procedure formally. For convenience, we denote the simulation procedure as $\textsf{Sim-Haar}$. Formally, $\textsf{Sim-Haar}$ takes as input a description of a circuit $C^{(\cdot)}$ that makes queries to $\mathcal{U}$ (and possibly some other oracle $\mathcal{Q}$), as well as some other parameters that we will describe shortly. 
It outputs a circuit $C'$ that does not make queries to $\mathcal{U}$ (but possibly to $\mathcal{Q})$. 
Here is a formal description.
\begin{algorithm}
\caption{$\textsf{Sim-Haar}^{(\cdot)}$}\label{sim-haar}
\vspace{2mm}
\textbf{Oracle access:} The algorithm has query access to an oracle $\mathcal{U}= \{\mathcal{U_{\ell}}\}_{\ell \in \mathbb{N}}$, where each $\mathcal{U_{\ell}}$ is a list of $2^{\ell}$ different $\ell$-qubit unitary transformations (and possibly to another oracle $\mathcal{Q}$).\\
\textbf{Input}: A quantum circuit $C^{(\cdot)}$ using space $s$ that makes $T$ queries to $\mathcal{U}$ (and possibly to another oracle $\mathcal{Q}$), $\eta \in \mathbb{N}$, and $\delta \in (0,1/3)$.\\
We denote the unitaries in the list $\mathcal{U_{\ell}}$ as $\{U_{k\ell} \}_{k\in \{0,1\}^{\ell}}$. Let $d=\log (192\frac{1}{\delta^2}(\eta+s)\cdot T^2+2)$.

\vspace{2mm}
For $l \in [s]$:
\begin{itemize}
\item If $\ell \in [d]$, view the list $\mathcal{U}_{\ell}$ as a unitary on a larger space that includes a control register for the index $k$. Classically simulate $\mathcal{U_{\ell}}$ by running the process tomography algorithm from Theorem~\ref{cor:processtom} on inputs $\epsilon = \frac{\delta}{T}$ and $\mu = \frac{1}{d} e^{-2(\eta+s)}$. This produces estimates $\Tilde{{\mathcal{U_{\ell}}}}$ such that $\| \Tilde{\mathcal{U_{\ell}}}(\cdot) \Tilde{{\mathcal{U_{\ell}}}}^{\dag}-{\mathcal{U_{\ell}}}(\cdot) \mathcal{U_{\ell}}^{\dag} \|_{\diamond} \leq \frac{\delta}{T}$ with probability at least $1-\frac{1}{d}e^{-2(\eta+s)}$. 
From Theorem \ref{cor:processtom}, this can be done, for each $\ell$,  with a number of queries that is $O\left(\frac{2^{2\ell}T}{\delta}\cdot \log(\frac{1}{e^{-\eta}})\right)\leq O\left(\frac{2^{2d}T}{\delta}\log(\frac{1}{e^{-\eta}})\right)\leq O
(\frac{T^3\eta^2}{\delta^3})$ and running time $poly(\ell, T/\delta).$  
\item Otherwise, if $\ell \in [d+1, s]$, do the following. Let $\epsilon= \frac{\delta}{s2^{s}}$. Let $m: \mathbb{N} \rightarrow \mathbb{N}$ be as in Lemma~\ref{lem:tdescons}. Let $A$ be the ``unitary design sampler'' algorithm from Lemma \ref{lem:tdescons} with parameters $\epsilon, n = \ell$, and $t = T$. Sample $f_{\ell}: \{0,1\}^{\ell}\rightarrow \{0,1\}^{m(l)}$ from a $2T$-wise independent family of functions. Set $\Tilde{U}_{k \ell}= A(f_{\ell}(k))$.
\end{itemize}
\textbf{Output}: The quantum circuit $C'$ defined as follows. $C'$ is identical to $C$ except that, for each $\ell$ and $k \in \{0,1\}^\ell$, all queries to $\mathcal{U}_{k \ell}$ are replaced with a direct application of the unitary $\Tilde{U}_{k\ell}$ defined earlier. Thus $C'$ makes no queries to $\mathcal{U}$ (but still makes queries to $\mathcal{Q}$, if $C$ does).
\end{algorithm}

In the following lemma, $\mathcal{U}$ is as in Algorithm~\ref{sim-haar}, and $\mathcal{Q}$ is any other fixed oracle.
\begin{lemma}[\cite{kretschmer2021quantum}]
\label{lem:sim-haar}
 Let $\eta \in \mathbb{N}$ and $\delta \in (0,1/3)$. Let $\textsf{Sim-Haar}_{\eta, \delta}$ denote Algorithm \ref{sim-haar} where the inputs $\eta$ and $\delta$ are fixed. 
Let $C^{(\cdot)}$ be a binary-output quantum circuit that uses space $s$ and makes $T$ queries to $(\mathcal{U,Q})$. Then, $\textsf{Sim-Haar}_{\eta, \delta}^{\mathcal{U}}(C)$ runs in time $poly(\eta,s,T,\frac{1}{\delta})$, and, with probability at least $1-2e^{-\eta}$ over the randomness of $\textsf{Sim-Haar}_{\eta, \delta}$ and the sampling of $\mathcal{U}$, we have that, for all $x \in \{0,1\}^{n}$ (where $n$ is the length of inputs to $C$),
\[
\Big|\Pr[C'^{\mathcal{Q}}(\ket{x})=1] - \Pr[C^{\mathcal{U},\mathcal{Q}}(\ket{x})=1]\Big|\leq 3\delta+e^{-\frac{\eta}{2}} \,.
\]
\end{lemma}
\begin{proof}
Recall that $d=\log (192\frac{1}{\delta^2}(\eta +s) T^2+2)$, and $\ket{x}$ is a computational basis state that $C^{\mathcal{U},\mathcal{Q}}$ takes as input. We define the following sequence of ``hybrids''. These are probability distributions, where the first is the output distribution of circuit $C^{\mathcal{U},\mathcal{Q}}(\ket{x})$, and the last is the output distribution of circuit $C'^{\mathcal{Q}}(\ket{x})$. We show that each two consecutive distributions are close. Let $x \in \{0,1\}^n$.
    \begin{enumerate}
        \item $H_1$: 
        $C^{\mathcal{U},\mathcal{Q}}(\ket{x})$. 
        \item $H_2$: $C_2(\ket{x})$, where $C_2$ is identical to $C^{\mathcal{U},\mathcal{Q}}$ except that, for all $\ell \in [d+1,s]$, $\mathcal{U}_{\ell}$ is replaced by a freshly sampled family of $2^{\ell}$ Haar random unitaries. 
        \item $H_3$: $C_3(\ket{x})$, where $C_3$ is sampled as follows. Let $A$ be the algorithm from Lemma~\ref{lem:tdescons} that samples from a phase invariant $\epsilon$-approximate  unitary $T$-design, where $\epsilon=\frac{\delta}{s2^{s}}.$ For $\ell \in [d+1,s]$, sample a function $g_{\ell}: \{0,1\}^{\ell}\rightarrow \{0,1\}^{m(\ell)}$ uniformly at random. $C_3$ is identical to $C_2$ except that, for $k \in \{0,1\}^\ell$, we replace queries to $U_{k \ell}$ with queries to $A(g_{\ell}(k))$.
        \item $H_4$: $C_4(\ket{x})$, where $C_4$ is sampled in the same way as $C_3$ except that we replace $g_\ell$ (which was previously sampled uniformly at random) with $f_\ell: \{0,1\}^{\ell} \rightarrow \{0,1\}^{m(\ell)}$, sampled from a $2T$-wise independent function family.
        
        \item $H_5$: $C_5(\ket{x})$, where $C_5$ is sampled in the same way as $C_4$ except that, for $\ell \in [d]$, queries to $\mathcal{U_{\ell}}$ are replaced by queries to $\Tilde{\mathcal{U_{\ell}}}$, obtained via the process tomography algorithm from Theorem \ref{cor:processtom} with parameters $\epsilon = \frac{\delta}{T}$, and $\mu = \frac1d e^{-2(\eta+s)}$. Note that this circuit is exactly $C'^{Q}(\ket{x})$.
    \end{enumerate}
     Let $f(\mathcal{U})= \Pr[C^{\mathcal{U,Q}}(\ket{x})=1]$. Lemma \ref{lem:lipshitz} implies that this function is $2T$-Lipshitz. Invoking the strong concentration of the Haar measure in Theorem~\ref{thm:conc} with $t=\delta$ and $L=2T$, we have that, for any standard basis input $\ket{x}$,
 $$
\Pr_{\mathcal{U}, \mathcal{U}'}[|\Pr[C^{\mathcal{U,Q}}(\ket{x})=1] -  \Pr[C^{\mathcal{U}', \mathcal{Q}}(\ket{x})=1]|\geq \delta]\leq \exp({-\frac{(2^d-2)\delta^2}{24(2T)^2}})   \leq e^{-2(\eta+s)} \,, $$
where the last inequality follows from the definition of $d$. By an averaging argument and a straightforward calculation, the latter implies that, in fact,
 $$
\Pr_{\mathcal{U}}\Big[\big|\Pr[C^{\mathcal{U,Q}}(\ket{x})=1] -  \Pr[C_2(\ket{x})=1]\big|\geq \delta + e^{-\eta-s}\Big]\leq  e^{-\eta-s} \, $$
(where the main difference from the previous expression is that the probability over $\mathcal{U}'$ has been absorbed inside $C_2$).

 
Now, $C_3$ replaces every unitary $U_{k \ell}$ for $\ell \in [d+1,s]$ and $k\in \{0,1\}^{\ell}$ with $A(g_{\ell}(k))$. Using Lemma~\ref{lem:tdesign}, the total change in acceptance probability is
 \[
 \sum_{\ell=d+1}^{s}\sum_{k \in \{0,1\}^{\ell}}\frac{\delta}{s2^{s}}\leq \delta \,.
 \]
 Thus, 
 \[
 \big|\Pr[C_3(\ket{x})=1] -\Pr[C_2(\ket{x})=1]\big|\leq \delta \,.
 \]
From \cite{Zha12}, we know that $T$ queries to a random function are perfectly indistinguishable from queries to a $2T$-wise independent family of functions. Thus, we have
 \[
 \Pr[C_4(\ket{x})=1]=\Pr[C_3(\ket{x})=1] \,.
 \]
From Theorem~\ref{cor:processtom}, for each $\ell \in [d]$, with probability at least $1- \frac{1}{d} e^{-2(\eta+s)}$ (over the randomness of the process tomography algorithm), we have that $\|\Tilde{\mathcal{U_{\ell}}}(\cdot) \Tilde{{\mathcal{U_{\ell}}}}^{\dag}-{\mathcal{U_{\ell}}}(\cdot )\mathcal{U_{\ell}}^{\dag} \|_{\diamond}\leq \frac{\delta}{T}$. Thus, by a union bound, with probability $\geq 1- e^{-2(\eta+s)}$ (over the randomness of the process tomography), we have that, for all $\ell \in [d]$,
$$\|\Tilde{\mathcal{U_{\ell}}}(\cdot) \Tilde{{\mathcal{U_{\ell}}}}^{\dag}-{\mathcal{U_{\ell}}}(\cdot )\mathcal{U_{\ell}}^{\dag} \|_{\diamond}\leq \frac{\delta}{T}\,.$$
Since $C_5$ and $C_4$ only make $T$ queries to $\mathcal{U}$, it follows, by triangle inequalities and Fact~\ref{fact}, that 
$$ \big| \Pr[C_5(\ket{x})=1]-\Pr[C_4(\ket{x})=1] \big| \leq \delta \,.$$
Adding up differences in acceptance probabilities (and adding up the probability losses) we get that, with probability at least $1-(e^{-\eta-s} + e^{-2(\eta+s)})$ over the randomness of $\mathcal{U}$, and the randomness in the process tomography (i.e. the randomness of $\textsf{Sim-Haar}_{\eta,\delta}$), 
$$ \Big|  \Pr[C^{\mathcal{U,Q}}(\ket{x})=1]-  \Pr[C'^{\mathcal{Q}}(\ket{x})=1] \Big|\leq 3 \delta + e^{-\frac{\eta}{2}} \,.
$$
Finally, taking a union bound over all standard basis inputs $\ket{x}$, we have that, with probability at least $1-2^s \cdot (e^{-\eta-s} + e^{-2(\eta+s)}) \geq 1 - 2 e^{-\eta} $ over the randomness of $\mathcal{U}$, and the randomness of $\textsf{Sim-Haar}_{\eta,\delta}$), for all standard basis inputs $\ket{x}$,
\begin{equation}
\label{eq:135}
\Big|  \Pr[C^{\mathcal{U,Q}}(\ket{x})=1]-  \Pr[C'^{\mathcal{Q}}(\ket{x})=1] \Big|\leq 3 \delta + e^{-\frac{\eta}{2}} \,,
\end{equation}
as desired.
\end{proof}

\subsection{An adversary breaking any QDS scheme relative to the oracle}
In this section, we prove Theorem \ref{thm:no-qds}. 
Concretely, we describe an adversary that, relative to $(\mathcal{U}, \mathcal{Q})$ (where $(\mathcal{U}, \mathcal{Q})$ is defined at the start of Section \ref{sec:sim-haar}), breaks any QDS scheme for messages of length $\ell(\lambda)$ for any $\ell$ such that $\ell(\lambda) \geq c \cdot \log(\lambda)$ for large enough $\lambda$. We show, for example, that one can take $c = 2$ (although our analysis is not tight).

We describe our adversary in Section~\ref{sec:adversary2}, and we provide the analysis in Section~\ref{sec:analysis2}. For a more informal overview see the technical overview (Sections \ref{sec:techoverview1} and \ref{sec:techoverview2}). 

\subsubsection{The adversary}
\label{sec:adversary2}

Let $(SKGen^{\mathcal{U, Q}}, \pkgen^{\mathcal{U, Q}}, \sign^{\mathcal{U, Q}}, Verify^{\mathcal{U, Q}})$ be a QDS scheme. We take the length of the secret key to be the security parameter $\lambda$.

$\mathcal{A}$ behaves as follows on input $1^{\lambda}$ (technically $\mathcal{A}$ also receives polynomially many copies of $\ket{pk}$, but it does not need them).
\begin{enumerate}
    \item Let $t= 40 \lambda$. $\mathcal{A}$ samples messages $m_1, \dots, m_t\leftarrow \mathcal{M}_{\lambda}$, and queries the challenger at these messages. The messages are sampled uniformly at random (possibly with repetitions) subject to the condition that $\bigcup_{i \in [t]} \{m_i\} \neq \mathcal{M}$. Let $\sigma_1,\dots, \sigma_t$ be the signatures returned by the challenger.
    \item Define the circuit $VerPKGen^{(\cdot)}(\cdot, \cdot, \cdot)$ to be such that 
    $$VerPKGen^{\mathcal{U, Q}}(sk,m,\sigma) := Verify^{\mathcal{U,Q}}(\pkgen(sk), m, \sigma)\,.$$ $\mathcal{A}$ obtains $VerPKGen'^{(\cdot)}(\cdot, \cdot, \cdot) \leftarrow \textsf{Sim-Haar}_{\lambda,\frac{1}{300}-e^{-\frac{\lambda}{2}}}^{\mathcal{U}, \mathcal{Q}}(VerPKGen)$, where $VerPKGen'^{(\cdot)}(\cdot, \cdot, \cdot)$ is a circuit that makes queries to $\mathcal{Q}$ (but not to $\mathcal{U}$). Here, as earlier, the notation $\textsf{Sim-Haar}_{\eta, \delta}$ refers to running algorithm $\textsf{Sim-Haar}$ (from Algorithm~\ref{sim-haar}) on the fixed inputs $\eta$ and $\delta$. Going forward, for ease of notation, we simply denote this circuit by $VerPKGen'$.
    \item Define the circuit $VerPKGenSign^{(\cdot)}(\cdot, \cdot, \cdot)$ to be such that 
    $$VerPKGenSign^{\mathcal{U, Q}}(sk,m, sk') := Verify^{\mathcal{U, Q}}(\pkgen(sk), m, \sign(sk',m)) \,.$$
    $\mathcal{A}$ obtains $VerPKGenSign'^{(\cdot)}(\cdot, \cdot, \cdot) \leftarrow \textsf{Sim-Haar}_{\lambda,\frac{1}{300}-e^{-\frac{\lambda}{2}}}^{\mathcal{U}, \mathcal{Q}}(VerPKGenSign)$. Note that $VerPKGenSign'^{(\cdot)}(\cdot, \cdot, \cdot)$ is a circuit that makes queries to $\mathcal{Q}$ (but not to $\mathcal{U}$). Going forward, for ease of notation, we simply denote this circuit by $VerPKGenSign'$.
 \item At this point, $\mathcal{A}$ invokes $\mathcal{Q}$. For simplicity, we will describe $\mathcal{Q}$'s behaviour as a probabilistic exponential time algorithm. However, formally, $\mathcal{Q}$ is a \emph{deterministic} function that, on input the instance of a fixed \textsf{EXP}-complete search problem, returns the solution. So, formally,
 \begin{itemize}
 \item $\mathcal{A}$ also provides the randomness as input to $\mathcal{Q}$ (and this is fine since the algorithm we describe only uses a polynomial length random string).
 \item $\mathcal{A}$ first computes a reduction from the search problem $P$ solved by the algorithm to the fixed \textsf{EXP}-complete problem, and maps the original input to the corresponding input according to the reduction.
 \end{itemize} 
From here on, we will not consider these two formalities. 

$\mathcal{A}$ provides $\{m_i, \sigma_i\}_{i \in [t]}$, $VerPKGen'$, $VerPKGenSign'$ as input to $\mathcal{Q}$, which returns a set $\mathsf{candidates}$ as in Algorithm \ref{algo:findsk} below.
\newpage
\begin{algorithm}
\caption{}
    \label{algo:findsk}
    \vspace{2mm}
    \textbf{Input}: $\{m_i, \sigma_i\}_{i \in [t]}$, $VerPKGen'$, $VerPKGenSign'$. 

    \begin{itemize}
    \item[1.] Initialize $\mathsf{Consistent} = \emptyset$. For $sk \in \{0,1\}^{\lambda}$:
    \begin{itemize}
    \item If $\Pr[VerPKGen'(sk, m_i, \sigma_i) = 1] \geq 9/10$ for all $i \in [t]$, update $\mathsf{Consistent} \leftarrow \mathsf{Consistent} \cup \{sk\}$\footnotemark\footnotetext{Recall that $VerPKGenSign'$ makes queries to $\mathcal{Q}$. Nonetheless, the problem of computing whether $\Pr[VerPKGenSign'(sk',m_i, \sigma_i)=1] \geq \frac{9}{10}$ can still be cast as an EXP problem: this is the problem of computing whether the magnitude squared of a particular entry of a vector, obtained by performing (exponentially-sized) matrices-vector multiplications, is $\geq\frac{9}{10}$. The problem can be cast in this way because each query to $\mathcal{Q}$ that the algorithm makes is a multiplication by an (exponential-sized) unitary that corresponds to solving $\mathcal{Q}$'s EXP-complete problem (on inputs of a certain polynomial size). Note that computing whether an entry of the resulting vector is greater or equal to some \emph{rational number} can be done deterministically.}.
    \end{itemize}
    \item[2.] Let $S_1 = \mathsf{Consistent}$, and $\mathsf{candidates} = \emptyset$. For $j \in [\lambda^2]$:
    \begin{itemize}
    \item Initialize $stingy_j = \emptyset$. For $sk \in S_j$:
    \begin{itemize}
    \item Let $friends_{sk} = \emptyset$.
    \item For each $sk' \neq sk  \in S_j$, do the following: 
    \begin{itemize}
    \item Count the number of $m \in \mathcal{M}$ such that $\Pr[VerPKGenSign'(sk, m, sk')=1] > \frac{1}{10}$. If this is at least $\frac{1}{10} \cdot |\mathcal{M}|$, update $friends_{sk} \leftarrow friends_{sk} \cup \{sk'\}$.
    \end{itemize}
    \item If $|friends_{sk}| \leq \frac12 \cdot |S_j|$, then $stingy_j \leftarrow stingy_j \cup \{sk\}$.
    \end{itemize}
    \item Initialize $goodSigner_j = \emptyset$. For $sk \in S_j$:
    \begin{itemize}
    \item Let $accept_{sk, j} = \emptyset$. 
     \item For $sk' \neq sk \in S_j$, do the following:
    \begin{itemize}
        \item Count the number of $m \in \mathcal{M}$ such that $\Pr[VerPKGenSign'(sk',m, sk)=1] \geq \frac{1}{10}$ (note that the role of $sk$ and $sk'$ is flipped compared to the definition of $friends_{sk}$). If this is at least $\frac{1}{10} \cdot |\mathcal{M}|$, update $accept_{sk, j} \leftarrow accept_{sk, j} \cup \{sk'\}$.
    \end{itemize}
    \item If $accept_{sk, j} = S_j\setminus\{sk\}$, update $$goodSigner_j \leftarrow  goodSigner_j \cup \{sk\}\,.$$
    \end{itemize}
    \item Let $S_{j+1} = stingy_j \cap goodSigner_j$. If $S_{j+1}=\emptyset$, halt and output $\mathsf{candidates}$. Otherwise, sample $sk \leftarrow S_{j+1}$. Update $\mathsf{candidates} \leftarrow \mathsf{candidates} \cup \{sk\}$.
    \end{itemize}
\item[3.] Output $\mathsf{candidates}$.
\end{itemize}
\end{algorithm}
\item Let $\mathcal{M}_{\text{queried}} = \bigcup_{i \in [t]} \{m_i\}$. $\mathcal{A}$ samples $sk \leftarrow \mathsf{candidates}$, $m \leftarrow \mathcal{M}\setminus \mathcal{M}_{\text{queried}}$, and runs \\$\sigma \leftarrow \sign^\mathcal{U, Q}(sk, m)$. $\mathcal{A}$ outputs $(m,\sigma)$.
\end{enumerate}

\subsubsection{The analysis}
\label{sec:analysis2}
Our adversary from Section \ref{sec:adversary2} invokes the procedure \textsf{Sim-Haar} from Section \ref{sec:sim-haar} twice: the first time, to obtain the circuit $VerPKGen' \gets \textsf{Sim-Haar}_{\lambda,\frac{1}{300}-e^{-\frac{\lambda}{2}}}^{\mathcal{U,Q}}(VerPKGen)$; the second time, to obtain the circuit $VerPKGenSign' \gets \textsf{Sim-Haar}_{\lambda,\frac{1}{300}-e^{-\frac{\lambda}{2}}}^{\mathcal{U,Q}}(VerPKGenSign)$. 

Recall that $VerPKGen'$ and $VerPKGenSign'$ are circuits that make queries to $\mathcal{Q}$ (but not to $\mathcal{U}$). Since $\mathcal{Q}$ is fixed, we will omit writing $\mathcal{Q}$ for the rest of the section, for ease of notation. By Lemma \ref{lem:sim-haar}, we have that, with probability at least $1-4\cdot e^{-\lambda}$ over $\mathcal{U}$ and the randomness of the two executions of \textsf{Sim-Haar}, it holds that for all $sk,sk', m,\sigma$:
\begin{align}
&\Big|\Pr[VerPKGen'(sk,m,\sigma) = 1] - \Pr[VerPKGen^{\mathcal{U}}(sk,m,\sigma) = 1] \Big| \leq \frac{1}{100} \quad \textnormal{ and } \label{eq:16} \\
&\Big| \Pr[VerPKGenSign'(sk,m,sk') = 1] - \Pr[VerPKGenSign^{\mathcal{U}}(sk,m,sk') = 1] \Big| \leq \frac{1}{100} \,. \label{eq:17}
\end{align}

For the rest of the analysis, we fix a $\mathcal{U}$ and circuits $VerPKGen'$ and $VerPKGenSign'$ such that \eqref{eq:16} and \eqref{eq:17} hold.

Now, let $sk^*$ denote the true secret key sampled by the generation procedure. For $sk \in \{0,1\}^{\lambda}$, let 
\begin{equation*}
p_{sk^*,sk} := \Pr_{m \leftarrow \mathcal{M}}\big[VerPKGen'(sk, m, \sign^{\mathcal{U}}(m, sk^*))=1\big]
\end{equation*}
By construction, if $sk \in \mathsf{Consistent}$, then $\Pr[VerPKGen'(sk,m_i,\sigma_i)] \geq \frac{9}{10}$ for all $i \in [t]$ (where the $(m_i, \sigma_i)$ are the message, signature pairs obtained by $\mathcal{A}$).

First, we show that, with overwhelming probability over the adversary's queries to the challenger, the set $\mathsf{Consistent}$ does not contain any $sk$ such that $p_{sk^*, sk} < \frac14$. Formally, let $sk$ be such that $p_{sk^*,sk} < \frac14$. Then, by an averaging argument, there must be a fraction $> \frac23$ of $m \in \mathcal{M}$ such that 
\begin{equation}
\label{eq:120}
\Pr[VerPKGen'(sk,m,\sign^{\mathcal{U}}(m, sk^*))=1] < 3/4 \,.
\end{equation}
By a further averaging argument, for those $m$ such that Equation \eqref{eq:120} holds, we have that with probability at least $1/9$ over $\sigma \leftarrow \sign^{\mathcal{U}}(m, sk^*)$, it must be that 
\begin{equation}
\Pr[VerPKGen'(sk,m,\sigma) = 1] < 9/10 \,.
\end{equation}
Thus, overall, if $p_{sk^*, sk} < \frac14$, then with probability at least $\frac23 \cdot \frac19 = \frac{2}{27}$ over $m \leftarrow \mathcal{M}$ and $\sigma \leftarrow \sign^\mathcal{U}(sk^*, m)$, we have
\begin{equation}
\Pr[VerPKGen'(sk,m,\sigma) = 1] < 9/10 \,.
\end{equation}
The above implies that, for any subset of size $\mathcal{M}'\subseteq \mathcal{M}_{\lambda}$ of size $|\mathcal{M}_{\lambda}|-1$, we have that, with probability at least $\frac13 \cdot \frac19 = \frac{1}{27}$ over $m \gets \mathcal{M}'$ and $\sigma \leftarrow \sign^\mathcal{U}(sk^*, m)$,\footnote{This last step is a bit of technicality. It is to deal with the fact that $\mathcal{A}$ samples the messages $m_i$ subject to the constraint that their union is not the whole of $\mathcal{M}$. The probability loss that we incur here (from $\frac23$ to $\frac13$) is a bit loose, and happens only when $\mathcal{M}_{\lambda}$ is very small (e.g.\ small constant size). When $\mathcal{M}_{\lambda}$ is exponentially large in $\lambda$, the loss is much smaller.}
\begin{equation*}
\Pr[VerPKGen'(sk,m,\sigma) = 1] < 9/10 \,. \end{equation*}
Now, recall that the condition by which Algorithm \ref{algo:findsk} places $sk$ in the set $\mathsf{Consistent}$ is precisely that, for all $i \in [t]$,
$$\Pr[VerPKGen'(sk, m_i, \sigma_i) = 1] \geq 9/10 \,.$$
Thus, the above implies that, if $p_{sk^*, sk} < \frac14$,
\begin{equation*}
    \Pr\Big[sk \in \mathsf{Consistent}\Big] \leq \left(\frac{26}{27}\right)^{t} =  \left(\frac{26}{27}\right)^{40\lambda}  \leq 2^{-2\lambda} \,,
\end{equation*}
for sufficiently large $\lambda$. Note that, crucially, this step of the proof relies on the fact that $t$ is a sufficiently large polynomial (and thus this step does not go through in the setting of one-time security).

Now, by a union bound, 
\begin{equation}
\label{eq:14}
\Pr\left[\exists sk \in \mathsf{Consistent} \textnormal{ such that } p_{sk^*,sk} <  \frac14\right] \leq 2^\lambda \cdot 2^{-2\lambda} = 2^{-\lambda} \,. 
\end{equation}
From now on, we will restrict our analysis to the case where, for all $sk \in \mathsf{Consistent}$, $p_{sk^*,sk} \geq  \frac14$, which by Equation \eqref{eq:14} happens at least with probability $1- 2^{-\lambda}$. Recall that, by definition of $p_{sk^*,sk}$, this means that
\begin{equation}
\label{eq:105}
\Pr_{m \leftarrow \mathcal{M}}\big[VerPKGen'(sk, m, \sign^{\mathcal{U}}(m, sk^*))=1\big] \geq \frac14 \,.
\end{equation}
Crucially, notice that the above guarantee is in terms of the circuit $VerPKGen'$, and it involves the honest signing procedure relative to oracle $\mathcal{U}$. However, parts (ii) and (iii) of Algorithm \ref{algo:findsk} instead involve the circuit $VerPKGenSign'$ (and no oracle $\mathcal{U}$). For the analysis above to fit with the subsequent analysis that we will derive from parts (ii) and (iii), we instead require a lower bound on $\Pr_{m \leftarrow \mathcal{M}}\big[VerPKGenSign'(sk, m, sk^*)=1\big]$. We show how to obtain such a lower bound from \eqref{eq:105}, with a (not too large) constant loss.


Notice that we have
\begin{align}
& \Pr[VerPKGenSign^{\mathcal{U}}(sk,m, sk^*) = 1]\\
&=\Pr[VerPKGen^{\mathcal{U}}(sk,m, \sign^{\mathcal{U}}(m,sk^*)  ) = 1] \\ &= \sum_{\sigma} \Pr[\sign^{\mathcal{U}}(m,sk^*) = \sigma]\cdot  \Pr[VerPKGen^{\mathcal{U}}(sk,m, \sigma) = 1] \\
&\geq \sum_{\sigma} \Pr[\sign^{\mathcal{U}}(m,sk^*) = \sigma]\cdot \left(\Pr[VerPKGen'(sk,m, \sigma) = 1] - \frac{1}{100}\right) \\
&= \sum_{\sigma} \Pr[\sign^{\mathcal{U}}(m,sk^*) = \sigma]\cdot \Pr[VerPKGen'(sk,m, \sigma) = 1] - \frac{1}{100} \\
&= \Pr[VerPKGen'(sk,m, \sign^{\mathcal{U}}(m, sk^*)) = 1] - \frac{1}{100} \\
&\geq \frac14-\frac{1}{100} = \frac{6}{25} \label{eq:109}\,, 
\end{align}
where the first inequality is by Equation \eqref{eq:16}, and the second inequality is by Equation \eqref{eq:105}.

Combining \eqref{eq:109} with \eqref{eq:17} via a triangle inequality gives
\begin{equation}
    \Pr_{m \leftarrow \mathcal{M}}[VerPKGenSign'(sk,m, sk^*) = 1] \geq \frac{6}{25} - \frac{1}{100} = \frac{23}{100} \,.
\end{equation}


To summarize, what we have established so far is that, with probability at least $1-2^{-\lambda}-4 \cdot e^{-\lambda}$ over the adversary's queries to the challenger (as well as over $\mathcal{U}$ and the randomness in the executions of \textsf{Sim-Haar}), if $sk \in \textsf{Consistent}$, then 
\begin{equation}
\Pr_{m \leftarrow \mathcal{M}}[VerPKGenSign'(sk,m,sk^*) = 1] \geq \frac{23}{100} \,.
\end{equation}

By an averaging argument, there must exist at least a $\frac{1}{10}$ fraction of $m$'s in $\mathcal{M}_{\lambda}$ such that 
\begin{equation}
\label{eq:110}
\Pr[VerPKGenSign'(sk,m,sk^*) = 1] \geq \frac{1}{10}  \,.
\end{equation}

Next, we establish the following.
\begin{lemma}
There is a negligible function $\negl$ such that $sk^* \in \mathsf{Consistent}$ with probability $1-\negl(\lambda)$ over the sampling of $sk^*$, the adversary's queries and challenger's responses, the sampling of $\mathcal{U}$, and the randomness of the two executions of $\mathsf{Sim}$-$\mathsf{Haar}$.
\end{lemma}
\begin{proof}
By the correctness property of the QDS scheme (Definition \ref{def:qds}), we know that there exists a negligible function $\negl'$ such that, for all $\lambda$, the following holds with probability $1-\negl'(\lambda)$ over the sampling of $sk^*$: for all $m \in \mathcal{M}_{\lambda}$,
$$ \Pr[\textit{Verify}^{\mathcal{U}}(\textit{PKGen}(sk^*), Sign(sk^*, m)) = 1] \geq 1 - \negl'(\lambda) \,. $$
Equivalently, for all $m \in \mathcal{M}_{\lambda}$, 
$$ \Pr[\textit{VerPKGen}^{\mathcal{U}}(sk^*, m, Sign(sk^*, m)) = 1] \geq 1 - \negl'(\lambda) \,, $$
which, using Equation \eqref{eq:16}, implies that, for all $m \in \mathcal{M}_{\lambda}$,
$$ \Pr[\textit{VerPKGen}'(sk^*, m, Sign(sk^*, m)) = 1] \geq 1 - \negl'(\lambda)$$
(where the probability is also over the randomness of the signing procedure).

This immediately implies that, for all $m \in \mathcal{M}_{\lambda}$, with probability $1-\negl'(\lambda)$ over $\sigma \gets Sign(sk^*, m)$, 
$$ \Pr[\textit{VerPKGen}'(sk^*, m, \sigma) = 1] \geq 1 - \negl'(\lambda) \,.$$

By a union bound, this implies that with probability $1- t \cdot \negl'(\lambda)$ over $\mathcal{A}$'s queries and the challenger's responses $(m_i, \sigma_i)$, for all $i$,
$$ \Pr[\textit{VerPKGen}'(sk^*, m_i, \sigma_i) = 1] \geq 1 - \negl'(\lambda) \geq \frac{9}{10} \,,$$
for large enough $\lambda$. This implies that, for some other negligible function $\negl$, we have that $sk^* \in \mathsf{Consistent}$ with probability $1-\negl(\lambda)$ over the sampling of $sk^*$, the adversary's queries and challenger's responses, as well as the randomness in $\mathcal{U}$ and in the two executions of $\textsf{Sim-Haar}$.
\end{proof}

The following lemma two lemmas are a crucial part of the analysis. Both lemmas, as well as the remainder of the analysis, hold with probability at least $1-\negl(\lambda)$ over the sampling of $sk^*$, the adversary's queries and challenger's responses, as well as the randomness of $\mathcal{U}$ and in the two executions of $\textsf{Sim-Haar}$.
\begin{lemma}
\label{lem:goodsigner}
$sk^* \in goodSigner_j$ for all $j \in [\lambda^2]$.
\end{lemma}

\begin{proof}
Since, by definition, $S_j \subseteq \mathsf{Consistent}$ for all $j$, and every $sk \in \mathsf{Consistent}$ satisfies Equation \eqref{eq:110}, it follows that $accept_{sk^*, j} = S_j \setminus \{sk^*\}$ for all $j$. Therefore, by definition of the set $goodSigner_j$, $sk^* \in goodSigner_j$ for all $j$. Note that this holds even when $S_j = \{sk^*\}$.
\end{proof}

The next key step of the analysis is to show that the set $S_j$ ``shrinks'' very quickly, and hence it can only contain a single element or be empty at the last iteration ($j = \lambda^2$).

\begin{lemma}
\label{lem:S}
For all $j$, either $|S_{j+1}| \leq \frac{9}{10} |S_j|$ or $|S_j| = 1$.
\end{lemma}
\begin{proof}
Suppose $|S_j| >1$, and suppose for a contradiction that $|S_{j+1}| > \frac{9}{10} |S_j|$. Then since, by definition, $S_{j+1} = stingy_j \cap goodSigner_j$, in particular this implies that $|goodSigner_j|> \frac{9}{10} |S_j|$ and $|stingy_j| > \frac{9}{10} |S_j|$.

To help clarity in the argument, for secret keys $sk,sk'$, we define $P_{sk,sk'}$ to be the predicate ``there exist at least a $\frac{1}{10}$ fraction of $m$ such that $\Pr[VerPKGenSign'(sk',m,sk) = 1] \geq \frac{1}{10}$.''

Now, for any $sk \in goodSigner_j$, we have that, by construction, $accept_{sk, j}= S_j\setminus\{sk\}$. This means precisely that $P_{sk,sk'}$ is true for all $sk' \neq sk \in S_j$. Morevoer, since by assumption $|goodSigner_j|> \frac{9}{10} |S_j|$, then there must be at least $\frac{9}{10} \cdot |S_j| \cdot (|S_j|-1)$ pairs $(sk,sk')$ (with $sk \neq sk'$) such that $P_{sk, sk'}$ is true.


On the other hand, for any $sk \in stingy_j$, we have that, by construction, $|friends_{sk}| \leq \frac12 |S_j|$. This means precisely that there is at most a $\frac12$ fraction of $sk' \neq sk \in S_j$ such that $P_{sk', sk}$ is true. Since, by assumption, $|stingy_j|> \frac{9}{10} |S_j|$, it must be that the number of pairs $(sk,sk')$ (with $sk \neq sk'$) such that $P_{sk', sk}$ is true is at most $\frac{9}{10}\cdot \frac{|S_j| \cdot(|S_j|-1)}{2} + \frac{|S_j| (|S_j -1)}{10} <  \frac{9}{10} \cdot |S_j| \cdot (|S_j|-1) $. The order of $sk$ and $sk'$ is flipped in $P_{sk', sk}$ compared to earlier, but of course this does not matter, and, equivalently, we have that the number of pairs $(sk,sk')$ (with $sk \neq sk'$) such that $P_{sk, sk'}$ is true is $<\frac{9}{10} \cdot |S_j| \cdot (|S_j|-1) $.  This contradicts the earlier statement (note that there is no contradiction when  $|S_j| = 1$, since the RHS would be zero, but we assumed $|S_j| >1$).

\end{proof}

With Lemmas \ref{lem:goodsigner} and \ref{lem:S} in hand, it is straightforward to complete the proof.

\vspace{2mm}
\noindent \textbf{Case 1:} Suppose $sk^* \notin S_{j+1}$ for some $j$. Let $j$ denote the first such index. This choice of $j$ implies that $sk^* \in S_j$ and $sk^* \notin S_{j+1}$ (note that $sk^* \in S_1$, so this implication is valid). 

Recall that $S_{j+1} =  stingy_j \cap goodSigner_j$. Then, since $sk^* \in goodSigner_j$ by Lemma \ref{lem:goodsigner}, it must be that $sk^* \notin stingy_j$. This implies that $$|friends_{sk^*}| > \frac12 \cdot |S_j|\,.$$ 
By Equation \eqref{eq:16}, for all $m, sk$,
\begin{equation}
\label{eq:102}
\Pr[VerPKGenSign^\mathcal{U}(sk^*,m,sk)= 1] > \Pr[VerPKGenSign'(sk^*, m, sk) = 1] -\frac{1}{100} \,.
\end{equation}
Additionally, for any $sk \in friends_{sk^*}$, we know that, by definition,
$$\Big|\big\{m: \Pr[VerPKGenSign'(sk^*,m,sk) =1] \geq \frac{1}{10}\big\}\Big| \geq \frac{1}{10} \cdot |\mathcal{M}_{\lambda}| \,. $$
Denote the set on the LHS by $\mathcal{M}_{good}(sk)$. We have
\begin{align}
&\Pr[VerPKGenSign^{\mathcal{U}}(sk^*, m, sk) = 1 \,\,\land \,\, m \notin \bigcup_{i \in [t]} \{m_i\}: sk \leftarrow S_j, m \leftarrow \mathcal{M_{\lambda}}] \\
&\geq  \Pr[sk \in friends_{sk^*}: sk \leftarrow S_j] \cdot \Pr_{m \leftarrow \mathcal{M}_{\lambda}}[m \in \mathcal{M}_{good}(sk) \,\,\land \,\, m \notin \bigcup_{i \in [t]} \{m_i\} \,|\,sk \in friends_{sk^*}] \label{eq:124}\\
&\quad\cdot \Pr[VerPKGenSign^{\mathcal{U}}(sk^*, m, sk) = 1 | sk \in friends_{sk^*},\, m \in \mathcal{M}_{good}] \nonumber\\
&\geq \frac12 \cdot \frac{\frac{1}{10}|\mathcal{M}_{\lambda}|-t}{|\mathcal{M}_{\lambda}|}\cdot \left(\frac{1}{10}- \frac{1}{100}\right)  \\
&\geq \frac12 \cdot (\frac{1}{10} - \frac{t}{|\mathcal{M}_{\lambda}|}) \cdot \left(\frac{1}{10}- \frac{1}{100}\right)  \\
&\geq \frac12 \cdot \frac{9}{100 }\cdot \frac{9}{100}
= \frac{81}{20000} \,, 
\label{eq:119}
\end{align}
where in the last line we used the fact that messages have length $\geq 2 \log(\lambda)$ for large enough $\lambda$, which means that $|\mathcal{M}_{\lambda}| \geq \lambda^2 \geq 100 t$, for large enough $\lambda$ (recall that $t = 40\lambda$).

Finally, note that Algorithm \ref{algo:findsk} places an element $sk \leftarrow S_j$ in the set $\mathsf{candidates}$. Hence, $\mathcal{A}$ samples this element with probability at least $\frac{1}{\lambda^2}$ in step $5$. By Equation \eqref{eq:119}, $\mathcal{A}$ wins the unforgeability game with probability at least $\frac{1}{\lambda^2}\cdot \frac{81}{20000}$.

\vspace{2mm}
\noindent \textbf{Case 2:}
Suppose $sk^* \in S_j$ for all $j$. Now, suppose for a contradiction that, for all $j$, $|S_{j+1}| \leq \frac{9}{10} |S_j|$. Then, since $|S_1|\leq 2^{\lambda}$, we would have $S_{\lambda^2} \leq (\frac{9}{10})^{\lambda^2} \cdot 2^{\lambda} <1$ (for $\lambda >1$). This would imply $|S_{\lambda^2}| = \emptyset$, which contradicts the hypothesis that $sk^* \in S_j$ for all $j$. Thus, by Lemma \ref{lem:S}, it must be that $|S_j| = 1$ for some $j$, and thus $S_j = \{sk^*\}$. This implies that, by the end of the algorithm, $sk^* \in \mathsf{candidates}$. Hence, $\mathcal{A}$ will sample $sk^*$ with probability at least $\frac{1}{\lambda^2}$ in step $5$, and thus $\mathcal{A}$ will win the unforgeability game with probability at least $\frac{1}{\lambda^2}\cdot (1-\negl'(\lambda))$, for some negligible function $\negl'$, by the correctness property of the QDS scheme.

\begin{remark}
\label{rem:technicality1}
When the message space is very small, e.g.\ constant size, the attack does not work as is. The reason is that, at a high level, the adversary is not able to gain very much from the polynomially many queries that it can make to the challenger. For example, if $\mathcal{M}_{\lambda} = \{0,1\}$, i.e.\ the message space is just a single bit, then the adversary can only ask for (polynomially many) signatures of a \emph{single} message, but not of the other one (since the unforgeability game requires the adversary to produce a signature of an unqueried message). In the setting of a large $\mathcal{M}_{\lambda}$, the adversaries' polynomially many queries are sufficient to guarantee, with high statistical confidence, that any secret key that is included in the set $\mathsf{Consistent}$, i.e.\ it is ``consistent'' with all of the $(m_i, \sigma_i)$ pairs, will also be ``consistent'' with freshly sampled message-signature pairs. In the case of a $\mathcal{M}_{\lambda} = \{0,1\}$ such a statistical guarantee does not hold. The specific step of the proof that makes use of the size of $\mathcal{M}_{\lambda}$ is Equation \eqref{eq:119}. 

As mentioned earlier, we can still rule out QDS with a small-size $\mathcal{M}_{\lambda}$ with a slightly different unforgeability definition than Definition \ref{def:polytimesec}. In this definition, for the adversary to succeed, it suffices to produce a valid signature that has not been previously produced by the challenger (even if this is the signature of a previously queried message). In this setting, since the adversary can succeed by signing a message that has been signed before, the polynomially many $(m_i, \sigma_i)$ that the adversary obtains provide a useful statistical guarantee, even if all of the $m_i$ are the same! By dropping the requirement that the adversary should produce a signature of a previously unseen message (which affects Equation \eqref{eq:124}), our current proof goes through essentially unchanged.
\end{remark}

\subsection{Why our attack does not work for QDS with quantum secret key and/or signatures}
\label{sec:barrier}
One of the main questions left open by this work is whether there exists a black-box construction from PRS of a QDS scheme with quantum public keys and \emph{quantum} secret keys and/or signatures.
Our attack from the previous subsection crucially only applies to QDS schemes with quantum public key but \emph{classical} secret key and signatures. The main issue in extending this attack is a bit subtle. The issue is that in Lemma \ref{lem:sim-haar}, which is based on the strong concentration of the Haar measure, the closeness guarantee of Lemma \ref{lem:sim-haar} holds for all \emph{standard basis} inputs to the circuits. Crucially, it does not hold for all possible quantum state inputs. This is because, the proof of Lemma \ref{lem:sim-haar} relies on a union bound, over all standard basis inputs, to obtain Equation \eqref{eq:135}. This gives a useful bound because the set of standard basis inputs is of size ``only'' $2^n$, where $n$ is the input-size of the circuit. However, it is unclear how to argue similarly when the union bound is over quantum states, since there are infinitely many of them.
One could hope to define an $\epsilon$-net for the set of quantum states (for some sufficiently small $\epsilon$), and apply a union bound over the states in the $\epsilon$-net. Unfortunately, the number of states in the $\epsilon$-net is doubly exponential (even when $\epsilon$ is a constant)! This is too large for a union bound to provide a non-trivial bound. So it seems that the strong concentration of the Haar measure is not quite strong enough for the current simulation technique to be useful in this setting. 

\begin{remark}
\label{rem:technicality2}
As mentioned earlier, our oracle separation does extend to a restricted kind of QDS scheme with \emph{quantum} secret keys. Specifically, if the QDS scheme is such that $\textit{SKGen}$ outputs a quantum state, but does \emph{not} query $\mathcal{U}$, then (a slight variation on) our attack still works. The point is that if $\textit{SKGen}$ does not query $\mathcal{U}$ at all, then the set of possible secret keys $\ket{sk}$ is ``only'' of size $2^{poly(\lambda)}$ (rather than being doubly exponential in an $\epsilon$-net): these are all of the post-measurement states that can result from running a poly-size quantum circuit and making a partial measurement on a subset of the qubits. Thus, a union bound is indeed still useful, and one can adjust the failure probability of $\mathsf{Sim}$-$\mathsf{Haar}$ appropriately (while maintaining polynomial runtime) based on the size of the circuit for $\textit{SKGen}$. 
\end{remark}

\newpage
\bibliographystyle{alpha}
\bibliography{references}

\begin{thebibliography}{AGQY23}

\bibitem[AGM21]{alagic2021can}
Gorjan Alagic, Tommaso Gagliardoni, and Christian Majenz.
\newblock Can you sign a quantum state?
\newblock {\em Quantum}, 5:603, 2021.

\bibitem[AGQY23]{ananth2023pseudorandom}
Prabhanjan Ananth, Aditya Gulati, Luowen Qian, and Henry Yuen.
\newblock Pseudorandom (function-like) quantum state generators: New
  definitions and applications.
\newblock In {\em Theory of Cryptography: 20th International Conference, TCC
  2022, Chicago, IL, USA, November 7--10, 2022, Proceedings, Part I}, pages
  237--265. Springer, 2023.

\bibitem[ALY23]{ananth2023pseudorandomStrings}
Prabhanjan Ananth, Yao-Ting Lin, and Henry Yuen.
\newblock Pseudorandom strings from pseudorandom quantum states.
\newblock {\em arXiv preprint arXiv:2306.05613}, 2023.

\bibitem[AQY22]{ananth2022cryptography}
Prabhanjan Ananth, Luowen Qian, and Henry Yuen.
\newblock Cryptography from pseudorandom quantum states.
\newblock In {\em Advances in Cryptology--CRYPTO 2022: 42nd Annual
  International Cryptology Conference, CRYPTO 2022, Santa Barbara, CA, USA,
  August 15--18, 2022, Proceedings, Part I}, pages 208--236. Springer, 2022.

\bibitem[BCQ22]{brakerski2022computational}
Zvika Brakerski, Ran Canetti, and Luowen Qian.
\newblock On the computational hardness needed for quantum cryptography.
\newblock {\em arXiv preprint arXiv:2209.04101}, 2022.

\bibitem[BHH16]{BHH}
G.~S. L.~Fernando Brandao, Adam~W. Harrow, and Michael Horodecki.
\newblock Local random quantum circuits are approximate polynomial-designs.
\newblock In {\em Communications in Mathematical Physics}, pages 397--434,
  2016.

\bibitem[CCS24]{chen2024power}
Boyang Chen, Andrea Coladangelo, and Or~Sattath.
\newblock The power of a single haar random state: constructing and separating
  quantum pseudorandomness.
\newblock {\em To appear}, 2024.

\bibitem[HKOT23]{HKOT23}
Jeongwan Haah, Robin Kothari, Ryan O’Donnell, and Ewin Tang.
\newblock Query-optimal estimation of unitary channels in diamond distance.
\newblock {\em arxiv preprint arxiv:2302.14066}, 2023.

\bibitem[Imp95]{impagliazzo1995personal}
Russell Impagliazzo.
\newblock A personal view of average-case complexity.
\newblock In {\em Proceedings of Structure in Complexity Theory. Tenth Annual
  IEEE Conference}, pages 134--147. IEEE, 1995.

\bibitem[IR89]{impagliazzo1989limits}
Russell Impagliazzo and Steven Rudich.
\newblock Limits on the provable consequences of one-way permutations.
\newblock In {\em Proceedings of the twenty-first annual ACM symposium on
  Theory of computing}, pages 44--61, 1989.

\bibitem[JLS18a]{ji2018pseudorandom}
Zhengfeng Ji, Yi-Kai Liu, and Fang Song.
\newblock Pseudorandom quantum states.
\newblock In {\em Advances in Cryptology--CRYPTO 2018: 38th Annual
  International Cryptology Conference, Santa Barbara, CA, USA, August 19--23,
  2018, Proceedings, Part III 38}, pages 126--152. Springer, 2018.

\bibitem[JLS18b]{JLS2018}
Zhengfeng Ji, Yi-Kai Liu, and Fang Song.
\newblock Pseudorandom quantum states.
\newblock In {\em Advances in Cryptology – CRYPTO 2018}, pages 126--152.
  Springer, 2018.

\bibitem[Kre21]{kretschmer2021quantum}
William Kretschmer.
\newblock Quantum pseudorandomness and classical complexity.
\newblock In {\em 16th Conference on the Theory of Quantum Computation,
  Communication and Cryptography}, 2021.

\bibitem[KT23]{khurana2023commitments}
Dakshita Khurana and Kabir Tomer.
\newblock Commitments from quantum one-wayness.
\newblock {\em arXiv preprint arXiv:2310.11526}, 2023.

\bibitem[Mec19]{Mec19}
Elizabeth~S. Meckes.
\newblock The random matrix theory of the classical compact groups.
\newblock {\em Cambridge Tracts in Mathematics}, pages 6--9, 2019.

\bibitem[MY22a]{morimae2022one}
Tomoyuki Morimae and Takashi Yamakawa.
\newblock One-wayness in quantum cryptography.
\newblock {\em arXiv preprint arXiv:2210.03394}, 2022.

\bibitem[MY22b]{morimae2022quantum}
Tomoyuki Morimae and Takashi Yamakawa.
\newblock Quantum commitments and signatures without one-way functions.
\newblock In {\em Advances in Cryptology--CRYPTO 2022: 42nd Annual
  International Cryptology Conference, CRYPTO 2022, Santa Barbara, CA, USA,
  August 15--18, 2022, Proceedings, Part I}, pages 269--295. Springer, 2022.

\bibitem[NC11]{NC}
Michael~A. Nielsen and Isaac~L. Chuang.
\newblock {\em Quantum Computation and Quantum Information: 10th Anniversary
  Edition}.
\newblock Cambridge University Press, USA, 10th edition, 2011.

\bibitem[RTV04]{reingold2004notions}
Omer Reingold, Luca Trevisan, and Salil Vadhan.
\newblock Notions of reducibility between cryptographic primitives.
\newblock In {\em Theory of Cryptography Conference}, pages 1--20. Springer,
  2004.

\bibitem[Zha12]{Zha12}
Mark Zhandry.
\newblock Secure identity-based encryption in the quantum random oracle model.
\newblock In {\em Advances in Cryptology--CRYPTO 2012: 32nd Annual
  International Cryptology Conference, CRYPTO 2012, Berlin, Heidelberg}, pages
  758--775. Springer, 2012.

\end{thebibliography}

\end{document}